\documentclass[11pt]{article}
\usepackage{amsmath,amsfonts,amssymb,amsthm}
\usepackage{graphicx}
\newcommand\mvector{\boldsymbol}

\newcommand\vZ{\mvector{Z}}

\newcommand\vgamma{\mvector{\gamma}}

\newcommand\field{\mathbb}

\newcommand\R{\field{R}}

\newcommand\C{\field{C}}
\newcommand\Z{\field{Z}}

\newcommand\N{\field{N}}

\renewcommand\Re{\operatorname{Re}}
\renewcommand\Im{\operatorname{Im}}

\newcommand\id{\operatorname{\mathrm{Id}}}
\newcommand\rmd{\mathrm{d}\mspace{1mu}}

\newcommand\rmi{\mathrm{i}\mspace{1mu}}
\newcommand\rme{\mathrm{e}}

\newcommand\pder[2]{\dfrac{\partial #1 }{\partial #2}} 
\newcommand\mtext[1]{\quad\text{#1}\quad}

\newtheorem{theorem}{Theorem}
\newtheorem{lemma}[theorem]{Lemma}
\newtheorem{proposition}[theorem]{Proposition}

\newtheorem{conjecture}{Conjecture}
\theoremstyle{definition}

\newtheoremstyle{note}{\topsep}{\topsep}{\slshape}{}{\scshape}{}{ }{}
\theoremstyle{note}

\newtheorem{remark}[theorem]{Remark}

\newtheorem{problem}[theorem]{Problem}
 \numberwithin{equation}{section}
 \numberwithin{theorem}{section}
\newcommand\sech{\operatorname{sech}\,}

\title{The Poisson equations in the nonholonomic Suslov problem: \\
Integrability, meromorphic and hypergeometric solutions}
\author{Yuri N. Fedorov \\
 Department de Matem\'atica Aplicada I, \\
Universitat Politecnica de Catalunya, \\
Barcelona, E-08028 Spain \\
e-mail: Yuri.Fedorov@upc.edu \\
and \\
Andrzej J. Maciejewski \\
Institute of Astronomy, University of Zielona G\'ora, \\
PL-65-246,  Zielona G\'ora, Poland \\
e-mail: maciejka@astro.ia.uz.zgora.pl \\ and \\
Maria Przybylska \\
Toru\'n Centre for Astronomy, Nicolaus Copernicus University \\
PL-87-100, Toru\'n, Poland \\ e-mail: mprzyb@astri.uni.torun.pl}

\begin{document}

\maketitle

\date{\small AMS Subject Classification  70F25; 37J60; 34M35;  70E40 }
\maketitle

\abstract{We consider the problem of integrability of the Poisson equations describing spatial motion of a rigid body in the classical nonholonomic Suslov problem. We obtain necessary conditions for their solutions to be meromorphic and
show that under some further restrictions these conditions are also sufficient. The latter lead to
a family of explicit meromorphic solutions, which
correspond to rather special motions of the body in space. We also give explicit extra polynomial integrals
in this case.

In the more general case (but under one restriction), the Poisson equations are transformed into a generalized
third order hypergeometric equation. A study of its monodromy group allows us also to calculate the ``scattering''
angle: the angle between the axes of limit permanent rotations of the body in space.}

\section{Introduction} In some cases of the rigid body dynamics, in particular, in the problem of motion of a solid
about a fixed point, the Euler equations for the angular velocity vector $\omega\in {\mathbb R}^3$ separate and can
be integrated. Then, given a generic solution $\omega(t)$, to determine the motion of the solid in space it is necessary
to solve the reconstruction problem, that is, to find 3 independent solutions of the linear Poisson equations
\begin{equation} \label{Poisson}
\dot \gamma= \gamma \times \omega(t),
\end{equation}
$\gamma\in {\mathbb R}^3$ being a unit vector fixed in space.

The most known example of solvable Poisson equations gives the Euler top problem, when generic $\omega(t)$
are elliptic (i.e., doubly periodic) functions and one particular solution $\gamma (t)$ is also elliptic, whereas the
other two are quasiperiodic (see, e.g., \cite{Jac1, Whitt}). 

Similar, but formally more complicated solutions $\gamma(t)$ appear in the case of the Zhukovsky--Volterra gyrostat
(see \cite{Volt, Zh}).

A nontrivial integrable generalization of the Euler--Poisson equations was found in \cite{BZ}, where the Euler
equations have the standard form and (\ref{Poisson})
are replaced by the equations
$$
\dot \gamma= \varkappa \, \gamma \times \omega(t),
$$
$\varkappa$ being an arbitrary {\it odd} integer number. It was shown that, like in the Euler top problem, the latter
equations possess an extra algebraic integral, however a complete solution for $\gamma$ is still unknown.

In present paper, following Suslov \cite{Su}, we consider the motion of the rigid body about a fixed point
in presence of constraint $\langle \omega, a\rangle=0$, $a$ being a fixed vector in the body frame.
Let $ {\mathbb I}\, : \,{\mathbb R}^3\mapsto {\mathbb R}^3$ be the symmetric
inertia tensor of the body. The Euler equations for the angular velocity vector $\omega$ separate and
take the following simple form
\begin{equation}
\frac {d}{dt}({\mathbb I} \omega)
={\mathbb I} \omega\times \omega+\lambda a,  \label{3.25}
\end{equation}
where $\times $ denotes the vector product in ${\mathbb R}^3$ and $\lambda$ is
the Lagrange multiplier.
Differentiating the constraint, we find
$\lambda =-\langle  {\mathbb I}\omega\times\omega,
{\mathbb I}^{-1} a\rangle / \langle a,{\mathbb I}^{-1} a\rangle . $
Therefore, (\ref{3.25}) can be represented as
$$
\frac{d}{dt} ({\mathbb I} \omega) = \frac 1{\langle a,{\mathbb I}^{-1} a\rangle }
{\mathbb I}^{-1} a  \times (({\mathbb I} \omega\times\omega)\times a) ,
$$
which, in view of $\langle \omega,a\rangle =0$, is equivalent to
\begin{equation}
\frac{d}{dt} ({\mathbb I}\omega)
= \langle {\mathbb I} \omega,a\rangle \, \omega\times {\mathbb I}^{-1} a.
\label{ep3.26}
\end{equation}
In the sequel, without loss of generality, we assume that $a=(0,0,1)^T$,
which, in view of the constraint, implies $\omega_3\equiv 0$. This simplifies the Poisson equations to the form
\begin{gather}
\label{eq:ps}
\dot\gamma_1=- \omega_2 (t) \gamma_3, \quad \dot\gamma_2= \omega_1(t) \gamma_3,
\quad \dot\gamma_3 = \omega_2 (t)\gamma_1 - \omega_1 (t) \gamma_2\, .
\end{gather}
We also assume that the tensor $\mathbb I$
is disbalanced, i.e., is not diagonal in the chosen frame. (If $a$ is an eigenvector of $\mathbb I$, then all the solutions
of (\ref{ep3.26}) are equilibria.)

It is known that under these assumptions the system (\ref{ep3.26}) restricted to the plane $\omega_3=0$
has a line of equilibria points $I_{13}\omega_1+ I_{23}\omega_2=0$ and that the trajectories
$\omega(t)$ are elliptic arcs that form the heteroclinic connection between the asymptotically unstable and stable
equilibria (see Fig. 1 below). 

Whereas the reduced system (\ref{ep3.26}) is elementary integrable in terms of hyperbolic functions,
it is believed that the corresponding Poisson equations (\ref{eq:ps}) are not,
although we did not find a proof of that in the literature. A study of complex solutions of these equations
was done in \cite{Kozlova}.
A qualitative analysis of the behavior of $\gamma(t)$ in the classical Suslov problem,
as well as in its multi-dimensional generalization was made in \cite{Bl_Z}, whereas
some other interesting generalizations of the problem were studied in \cite{Jo2}.

\paragraph{Contents of the paper.}
In Section 2 we present generic solutions of the Euler equations for the Suslov problem and formulate the problem
of integrability of the system (\ref{ep3.26}), (\ref{eq:ps}) by the Euler--Jacobi theorem.

In Section 3 necessary and sufficient conditions of meromorphicity of solutions are obtained, they both require that
one of the components $I_{13}, I_{23}$ of the inertia tensor must be zero.

Section 4 discusses general properties of solutions of the Poisson equations in the specific case $I_{13}=0$,
which are compared with generic solutions of another famous nonholonomic system, the Chaplygin sleigh. It is also
shown that the meromorphicity conditions on $\mathbb I$ are compartible with the restrictions on the inertia tensor
of a physical rigid body.    

In Section 5 we show that in the case $I_{13}=0$ the Poisson equations can be transformed to
a generalized third order hypergeometric equation.
Using its monodromy group, we solve the classical problem of calculating
the angle between the axes of limit permanent rotations of the body in space.

When the sufficient conditions of meromophicity are satisfied, we observe that the corresponding
hypergeometric series
solutions reduce to products of polynomials and exponents, which are explicitly calculated.

Next, we apply the differential Galois analysis to the Poisson equations and prove that when the parameters
of the problem satisfy the necessary conditions of meromophisity, but not the sufficient ones, these equations
and, therefore, the whole Suslov system, are not solvable in the class of Liouvillian functions.

In Section 6 we present all the meromorphic solutions of the problem and the corresponding extra polynomial
integrals in the explicit form.

In Conclusion some relevant open problems are briefly discussed.

\section{Generic solutions of the Euler and the Poisson equations}
In the general case the components
$I_{13}, I_{12}$ of the inertia tensor ${\mathbb I}$ are not zero, but,
by an appropriate choice of the frame that preserves the constraint one can always make
$ I_{12}=0$. For We also impose the normalisation $\det {\mathbb I}=1$.
Condition that ${\mathbb I}$ is positively defined means that all main minors are greater than zero and we obtain $I_{11}>0$ and $I_{11}I_{22}>0$ that gives $I_{22}>0$.
\medskip

Then the Euler-Poisson equations \eqref{ep3.26}, \eqref{eq:ps} have the form
\begin{equation}
\label{gen12}
\dot\omega_1= I_{22}( I_{13}\omega_1+ I_{23}\omega_2) \omega_2 ,
\quad \dot\omega_2=- I_{11}( I_{13}\omega_1+ I_{23}\omega_2) \omega_1 ,
\end{equation}
and
\begin{equation}
\begin{aligned}
 \dot \gamma_1 &= -\omega_2(t)\gamma_3, \\
\dot \gamma_2 &= \omega_1(t)\gamma_3, \\
\dot \gamma_3 &= \omega_2(t)\gamma_1 -  \omega_1(t)\gamma_2.
\end{aligned}
\label{eq:whole}
\end{equation}
They  have two first integrals, the energy and the trivial geometric one:
\begin{equation}
 F_1=I_{11}\omega_1^2+I_{22}\omega_2^2, \quad F_2=\langle\gamma,\gamma\rangle=\gamma_1^2+\gamma_2^2+\gamma_3^2.
\end{equation}

In this paper the main problem we consider is:
{\it For which values of $I_{i,j}$ the system~ \eqref{gen12}, \eqref{eq:whole} is integrable?}
Here the integrability will be understood in the context of the classical Euler--Jacobi theorem, which relies upon
the existence of an invariant measure and sufficient number of independent integrals.

Although the system~\eqref{gen12}, \eqref{eq:whole} does not have an invariant measure in the strict sense
(the density of the volume form 
tends to $0$ as $t\to\pm\infty$),
we can consider its restriction on the energy level $F_1= E>0$, which consists of two open components.
On each of them we choose time $t$ and $\gamma$ as coordinates and then the restricted equations read
\begin{equation}
\begin{split}
&\dot t=1,\\
 &\dot \gamma_1 = -\omega_2(t)\gamma_3, \\
&\dot \gamma_2 = \omega_1(t)\gamma_3, \\
&\dot \gamma_3 = \omega_2(t)\gamma_1 -  \omega_1(t)\gamma_2.
\end{split}
\label{eq:whole1}
\end{equation}
They have the trivial integral $F_2$, as well as the invariant volume form
\[
 \mu=\mathrm{d}t\wedge\mathrm{d}\gamma_1\wedge\mathrm{d}\gamma_2\wedge\mathrm{d}\gamma_3,
\]
Thus, for the integrability in the Jacobi sense only one additional first integral $F(t,\gamma)$ is required.

The existence of such an integral is closely related to the existence of single-valued solutions of the Poisson
equations. Namely, let $P(t)=(P_1,P_2,P_3)$ a single-valued vector-function satisfying
$\dot P=P\times\omega$. Then the system possesses the single-valued integral
\begin{equation}
 F= \langle P(t),\gamma \rangle,
\label{eq:liny}
\end{equation}
Indeed,
\[
 \dot F=\langle \dot P(t),\gamma \rangle+\langle P(t),\dot\gamma \rangle=
\langle \dot P(t),\gamma \rangle+\langle P(t),\gamma\times\omega \rangle
=\langle \dot P(t)+\omega\times P(t),\gamma\rangle=0\, .
\]
If, moreover, the components of $P(t)$ are single-valued functions of the solutions $\omega_1(t), \omega_2(t)$, then
the system~\eqref{gen12}, \eqref{eq:whole} admits an additional first integral
$$
 F_3(\omega_1,\omega_2, \gamma_1,\gamma_2,\gamma_3) =\langle P(\omega),\gamma \rangle
$$
functionally independent with $F_1$ and $F_2$.

We start with generic solution of the dynamic equations (\ref{gen12}), which have the form
\begin{equation}
\omega_{1}(t) =\frac{a (e^{At} - e^{-At}) +c_{1} }{e^{At}+e^{-At}}, \quad
\omega_{2} (t) =\frac{ b (e^{At}- e^{-At})  +c_{2}}{e^{At}+e^{-At}}, \quad b=-a \dfrac{I_{13}}{I_{23}}
\end{equation}
where
$$
a =
\frac {A I_{23} } { I_{13}^2 I_{22} + I_{23}^2 I_{11}} , \quad
c_{1}
= \pm 2  \frac{ A I_{13} }{ I_{13}^2 I_{22} + I_{23}^2 I_{11}} \sqrt{ \frac{I_{22}}{I_{11}} } ,\quad
c_{2}
= \pm 2  \frac{ A I_{23} }{ I_{13}^2 I_{22} + I_{23}^2 I_{11}} \sqrt{ \frac{I_{11}}{I_{22}} },
$$
$A$ being an arbitrary positive constant related to the energy integral.
Note that for $t\to \pm\infty$ these expressions give points on the
equilibria line $I_{13}\omega_1+ I_{23}\omega_2=0$,  as required.

\begin{remark} \label{A=1}
For each $A$ fixed, the Poisson equations \eqref{eq:whole} and the above functions $\omega(t)$ are invariant
with respect to time rescaling $t \to t/A$, which reduces the solution of (\ref{gen12}) to the form
\begin{gather} \label{eq:o12}
\omega_{1}(t) =\frac{a (e^{t} - e^{-t}) +c_{1} }{e^{t}+e^{-t}}, \quad
\omega_{2} (t) =\frac{ b (e^{t}- e^{-t})  +c_{2}}{e^{t}+e^{-t}}, \quad b=-a \dfrac{I_{13}}{I_{23}}, \\
a = \frac {I_{23} } { I_{13}^2 I_{22} + I_{23}^2 I_{11}} , \quad
c_{1}  = \frac{\pm 2   I_{13} }{ I_{13}^2 I_{22} + I_{23}^2 I_{11}} \sqrt{ \frac{I_{22}}{I_{11}} } ,\quad
c_{2} =\frac{  \pm 2  I_{23} }{ I_{13}^2 I_{22} + I_{23}^2 I_{11}} \sqrt{ \frac{I_{11}}{I_{22}} }.
\label{c's}
\end{gather}
This implies that, without loss of generality, one can study solutions of \eqref{eq:whole} with the coefficients
$\omega_1, \omega_2$ having the form \eqref{eq:o12}, \eqref{c's}. This will be assumed in the sequel.

Note that choosing $\omega(t)$ in the form \eqref{eq:o12} implies that the energy integral is fixed to be
\begin{equation}
 F_1=\frac{ 1 }{ I_{13}^2 I_{22} + I_{23}^2 I_{11}} \, .
\label{eq:enlev}
\end{equation}
\end{remark}

\begin{figure}[h,t]\label{traj_fig}
\begin{center}
\includegraphics[width=.5\textwidth]{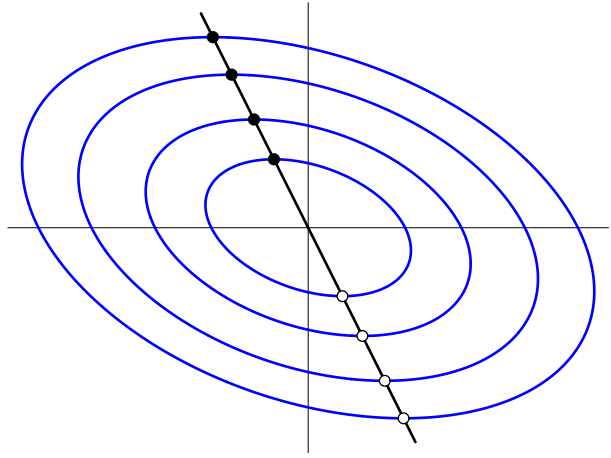}
\caption{\footnotesize Generic trajectory $\omega(t)$ on the plane $(\omega_1, \omega_2)$. Black and
white dots indicate stable and unstable equilibria.}
\end{center} \end{figure}

\section{Painlev\'e property}
In this section  we investigate the  following problem.
\begin{problem}
 For which values of the parameters $I_{ij}$
all solutions of the system~\eqref{eq:whole} are meromorphic in $\C$.
\end{problem}
The answer is contained in the following theorem.
\begin{theorem}
 \label{pro:necc}
All solutions of~\eqref{eq:whole} are meromorphic solutions iff either
\begin{align}
\label{eq:necc1}
I_{13} & =0,\mtext{and} \dfrac{I_{11}-I_{22}}{I_{11}I_{22}(I_{11}I_{22}I_{33}-1)} =p^2, \\
\mtext{or}
I_{23} & =0,\mtext{and}\dfrac{I_{11}-I_{22}}{I_{11}I_{22}(I_{11}I_{22}I_{33}-1)}
=-p^2, \label{eq:necc11}
\end{align}
where $p$ is a nonzero integer.
\end{theorem}

Below we also show that the second case~\eqref{eq:necc11} can be obtain from~\eqref{eq:necc1} by the linear transformation
\begin{equation}
\begin{split}
& \omega_1\rightarrow -\omega_2,\quad \omega_2\rightarrow -\omega_1, \quad \gamma_1\rightarrow\gamma_2,\quad
\gamma_2\rightarrow\gamma_1,\quad \gamma_3\rightarrow\gamma_3,\\
&I_{11}\rightarrow I_{22},\quad I_{22}\rightarrow I_{11},\quad I_{13}\rightarrow I_{23},\quad
I_{23}\rightarrow I_{13},\quad I_{33}\rightarrow I_{33}.
\end{split}
\label{eq:lina}
\end{equation}

Note that the conditions~\eqref{eq:necc1} expressed in terms of the parameters $(a,b,c_1,c_2)$ take the following form
\begin{equation}
 \label{eq:necc2}
c_1=b=0\mtext{and}c_2^2 -4a^2=4p^2,\quad p\in\Z^{\ast}.
\end{equation}

The proof of Theorem \ref{pro:necc} is given in two next subsections.

\subsection{Necessary conditions---Kovalevskaya analysis}
Let us show that if all solutions of the system~\eqref{eq:whole} are single-valued, then one of the
conditions~\eqref{eq:necc1}, \eqref{eq:necc11} is satisfied.

First observe that the right hand sides of the system~\eqref{gen12}, \eqref{eq:whole} are homogeneous polynomials of degree 2 in
the components of $x=(\omega_1,\omega_2,\gamma_1,\gamma_2,\gamma_3)$.
Then we can apply the Kovalevskaya-Lyapunov analysis (see~\cite{Kozlov:96::} for details) to find
that the system admits two solutions of the form
\begin{equation}
 x(t)=\frac{1}{t} {\bf d}, \mtext{where} {\bf d}=(\bar\omega_1, \bar \omega_2,\bar\gamma_1,\bar\gamma_2,\bar\gamma_3) .
\end{equation}
The first of them is given by
\begin{equation}
\label{eq:d}
\begin{split}
 \bar  \omega_1 &=\dfrac{1}{I_{11}I_{23}+\rmi I_{13}\sqrt{I_{11}I_{22}}}, \quad
\bar \omega_2=-\dfrac{1}{I_{13}I_{22}-\rmi I_{23}\sqrt{I_{11}I_{22}}},\\
\bar \gamma_1&= \bar \gamma_2= \bar \gamma_3=0,
\end{split}
\end{equation}
and the second one is its complex conjugate.
According to the Lyapunov theorem, if all solutions of~\eqref{eq:whole} are single-valued, then
\begin{enumerate}
 \item all eigenvalues of the Kovalevskaya matrix
$$
 K(d)=\pder{v}{x}(d)+\id,
$$
are integers, and
\item the Kovalevskaya matrix is semi-simple.
\end{enumerate}
For the solution $d$ given by~\eqref{eq:d} the characteristic polynomial of $K(d)$ is
\begin{equation}
  \det(K(d)-\lambda \id)=\dfrac{(\lambda-2) (\lambda-1 )
(\lambda+1)W(\lambda)}{I_{11} I_{22}(I_{11} I_{22} I_{33}-1)^2},
\end{equation}
where
\[
 W(\lambda)=( \sqrt{I_{11}} I_{23}-\rmi I_{13} \sqrt{I_{22}}) [I_{22} +
   I_{11} (-1 - I_{22} (I_{13} \sqrt{I_{22}} - \rmi \sqrt{I_{11}} I_{23})^2
(\lambda-1)^2)].
\]
The roots of $W(\lambda)$ have the form
\[
 \lambda_{1,2}=1\pm\Lambda,\quad\Lambda=\dfrac{\sqrt{(I_{11} -
I_{22})(\sqrt{I_{11}} I_{23}+\rmi I_{13} \sqrt{I_{22}})^2}\, (\sqrt{I_{11}}
I_{23}-\rmi I_{13} \sqrt{I_{22}})^2}{\sqrt{I_{11}  I_{22}}(I_{11} I_{22}
I_{33}-1)^2}.
\]
If $\lambda_1$ and $\lambda_2$ are integer, then, in particular, they are real, which implies that
\begin{equation}
\label{eq:neccg1}
I_{13}I_{23}(I_{11}-I_{22})=0.
\end{equation}
Moreover, if  $\lambda_1,\lambda_2\in\Z$, then  $\Lambda^2=p^2$, where $p\in\Z^\star$. This gives
\begin{equation}
\label{eq:neccg12}
\dfrac{(I_{11}-I_{22})(I_{11}I_{23}^2-I_{22}I_{13}^2)}{I_{11}I_{22}(I_{11}I_{22}I_{33}-1)^2}=p^2.
\end{equation}
The condition \eqref{eq:neccg1} gives three possibilities. Either  $I_{13}=0$, and this case gives~\eqref{eq:necc1}, or  $I_{23}=0$, which gives~\eqref{eq:necc11}, or, finally,   $I_{11} = I_{22}$.  In the last case, the Kovalevskaya matrix
has the multiple eigenvalue $\lambda=1$ and is not semi-simple.
As a result, we proved the `if' part of Theorem~\ref{pro:necc}.
\begin{remark}
  In terms parameters $(a,b,c_1,c_2)$ conditions~\eqref{eq:neccg1} and \eqref{eq:neccg12} have  the form
\begin{equation}
 \label{eq:neccg2}
ac_1+bc_2=0\mtext{and} c_1^2+c_2^2 -4(a^2+b^2)=4p^2,\quad p\in\Z.
\end{equation}
\end{remark}

\subsection{Sufficient conditions---analysis of the monodromy group}
The general solution~\eqref{eq:o12} of the Euler equations is single-valued.
Thus, all solutions of the Euler-Poisson equations~\eqref{gen12}, \eqref{eq:whole} are single-valued if and only if all the solutions of
the linear Poisson equations~\eqref{eq:whole} with $\omega_1(t)$ and $\omega_2(t)$ as in \eqref{eq:o12} are single-valued.

Note that the only singular points of ~\eqref{eq:whole} are simple poles located at
$t_0=\pm\frac{\pi}{2} \rmi \mod \pi\rmi$.
Then a necessary condition for the single-valuedness is that
the monodromy matrices at all the singular points are identities.

To show this, rewrite the system~\eqref{eq:whole} in the form
\begin{equation}
 \label{eq:at}
\dot \gamma=A(t)\gamma,\qquad
A(t)= \begin{bmatrix}
        0 & 0 &-\omega_2(t)\\
        0& 0 &\phantom{-}\omega_1(t)\\
      \omega_2(t) & -\omega_1(t)& 0
      \end{bmatrix}
\end{equation}
and observe that
\begin{equation}
\label{eq:am1}
 A(t)= \frac{1}{t-t_0} A_{0} +O(t-t_0), \quad
A_{0}=\begin{bmatrix}
        0 & 0 &- b+\frac{1}{2}\rmi c_2 \\[0.5em]
        0& 0 &\phantom{-} a+\frac{1}{2}\rmi c_1  \\[0.5em]
      b-\frac{1}{2}\rmi c_2  &-a+\frac{1}{2}\rmi c_1  & 0
      \end{bmatrix}.
\end{equation}

The monodromy matrix of the canonical loop around $t_0$ is $M_{t_0}:=\exp
2\pi\rmi A_{0}$.   A direct calculation shows that the eigenvalues of $A_{0}$ are
\begin{equation}
 \label{eq:ro}
\rho_1=0 \mtext{and} \rho_{2,3}=\pm\frac{1}{2}\sqrt{c_1^2+c_2^2
-4(a^2+b^2)+4\rmi (ac_1+bc_2)}.
\end{equation}
Hence, if the conditions~\eqref{eq:neccg2} are satisfied, then  $A_{0}$  has eigenvalues
$\rho_1=0$, $\rho_2=p$, and $\rho_3=-p$ with $p\in\Z^\star$. In this case
$M_{t_0}=\id$ provided that there is no local logarithmic solutions in a neighborhood of $t_0$.
In order to check this, we use the following theorem (see, e.g., \cite{Barkatou} for the details).

\begin{theorem}
\label{thm:bar}
Assume that in the linear system
\begin{equation}
 \dot Y=B(t)Y,\qquad B(t)=t^{-1}\sum_{i=0}^{\infty} B_it^i,
\label{eq:linb}
\end{equation}
the matrix coefficient $B_0$ does not have a pair of eigenvalues such that their difference
is a non-zero integer. Then there exists a matrix $T$ given by a convergent power series
\[
 T=\sum_{i=0}^{\infty}T_it^i,\qquad T_0=\id,
\]
such that the linear map  $Y=TZ$ transforms \eqref{eq:linb} into the following form
\begin{equation}
Z'=\frac{1}{t}B_0Z.
\label{eq:lincons}
\end{equation}
\end{theorem}
The form of the fundamental matrix $\vZ_{B_0}$ of \eqref{eq:lincons} is
well known. Namely, let $P$ be the similarity transformation reducing $B_0$ into
its Jordan form $J=P^{-1}B_0P$, where $P$ is an invertible matrix with constant
coefficients, and let $\vZ_{J}$ denote the new fundamental matrix. Then
$\vZ_{B_0}=P\vZ_{J}P^{-1}$, and $\vZ_{J}$ has a block-diagonal
structure such that the Jordan block (of dimension $\nu$) corresponding to an eigenvalue $\lambda$ has the form
\[
 t^{\lambda}\begin{bmatrix}
             1&\ln t& \cdots &\frac{\ln^{\nu-1}t}{(\nu-1)!}\\
0&1&\ddots& \vdots \\
\vdots&\ddots&\ddots&\ln t\\
0&\ldots&0&1
            \end{bmatrix} .
\]

Now assume that conditions~\eqref{eq:necc1} or ~\eqref{eq:necc11} are satisfied.
By the existence of the linear transformation \eqref {eq:lina},
without loss of the generality we can assume that conditions~\eqref{eq:necc1} are satisfied.

Expressed in terms of parameters $(a,b,c_1,c_2)$, the latter take the form \eqref{eq:necc2}, or, setting
 $c_2=-2ca$,  $c^2=I_{11}/I_{22}$, the following form
\begin{equation}
\label{eq:cpr}
 (c^2-1)a^2 = p^2.
\end{equation}
We also denote $d^2=c^2-1$.

Now we check the presence of logarithmic terms in the solutions.
\begin{theorem}
 The Poisson equations~\eqref{eq:whole} with $\omega_i$ given by \eqref{eq:o12}
with $c_1=b=0$, $c_2=-2ca=-2\sqrt{d^2+1}p/d$, $p\in\Z$,  do not have
logarithmic terms in local solutions around singular points $t_0$.
\end{theorem}
\begin{proof}
For $c_1=b=0$, $c_2=-2ca=-2\sqrt{d^2+1}p/d$, the matrix $A_{0}$ reads
\begin{equation}
 \label{eq:am2}
A_{0}=\begin{bmatrix}
        0 & 0 &-\dfrac{\rmi p\sqrt{d^2+1}}{d} \\[0.5em]
        0& 0 &\dfrac{p}{d}  \\[0.5em]
     \dfrac{\rmi p\sqrt{d^2+1}}{d}  &-\dfrac{p}{d}  & 0
      \end{bmatrix}.
\end{equation}
Using the similarity transformation $S$, one obtains its diagonal form: $\tilde
A_0=S^{-1}A_0S$ with
\[
 \tilde A_0=\begin{bmatrix}
             p&0&\phantom{-}0\\
0&0&\phantom{-}0\\
0&0&-p
            \end{bmatrix},\quad
S=\begin{bmatrix}
  -\dfrac{\rmi \sqrt{d^2+1}}{d}&-\dfrac{\rmi}{\sqrt{d^2+1}}&
\dfrac{\rmi \sqrt{d^2+1}}{d}\\
\dfrac{1}{d}&1&-\dfrac{1}{d}\\
1&0&1
  \end{bmatrix}.
\]
Thus we see that all the eigenvalues of $A_0$ are integer and their differences
are also nonzero integers. Under this transformation
the matrix $A$ of the linear Poisson equations takes the form
\[
  \tilde A=\begin{bmatrix}
            \dfrac{\rmi p\sech(t)(d^2+1+\rmi\sinh(t))}{d^2}&\dfrac{\rmi
p(\sech(t)+\rmi\tanh(t))}{2d}&0\\
\dfrac{(d^2+1)p(-\rmi\sech(t)+\tanh(t))}{d^3}&0&\dfrac{
(d^2+1)p(-\rmi\sech(t)+\tanh(t))}{d^3},\\
0&\dfrac{\rmi p(\sech(t)+\rmi\tanh(t))}{2d}&\dfrac{
p[-\rmi(d^2+1)\sech(t)+\tanh(t)]}{d^2}
           \end{bmatrix}
\]
Transformation $\gamma \rightarrow \Gamma$ defined by $\gamma=U\Gamma$ with
\[
 U=\begin{bmatrix}
    \left(t-\dfrac{\rmi \pi}{2}\right)^{2p}&0&0\\
0& \left(t-\dfrac{\rmi \pi}{2}\right)^{p}&0\\
0&0&1
   \end{bmatrix},
\]
gives the equivalent linear system $\dot\Gamma=B(t)\Gamma$ with
\[
 B(t)=U^{-1}AU-U^{-1}\dot U.
\]
This matrix has residual part diagonal
\[
 B_0=\begin{pmatrix}
      -p&\phantom{-}0&\phantom{-}0\\
\phantom{-}0&-p&\phantom{-}0\\
\phantom{-}0&\phantom{-}0&-p
     \end{pmatrix},
\]
which satisfies the assumptions of Theorem~\ref{thm:bar}. Since all the Jordan blocks of $B_0$ are
one-dimensional, the logarithmic terms do not appear in the solutions of the Poisson equations.
\end{proof}

\section{Special case of the inertia tensor. General properties of solutions}
\label{ssec:spec}

Assume now that, apart from ${\mathbb I}_{12}=0$, the first meromorphisity condition in \eqref{eq:necc1} is also
satisfied, that is, ${\mathbb I}_{13}=0$, which imposes an essential restriction on $\mathbb I$.

Then the reparametrized solutions \eqref{eq:o12}, \eqref{c's} are reduced to
\begin{gather}
\label{omegas}
 \omega_1(t) = \frac{a\left(\rme^t-\rme^{-t}\right)}{\rme^t+\rme^{-t}}, \qquad
 \omega_2(t) = -\frac{2ac}{\rme^t+\rme^{-t}}, \\
a= \frac {I_{22} I_{33}-I_{23}^2}{I_{23} } , \quad c= \sqrt{ \frac{I_{11}}{I_{22}} } .  \nonumber
\end{gather}
It is seen that they have equilibria along the line $\omega_2=0$.

Using the parameter $p$ introduced in \eqref{eq:necc1}, and setting $d=p/a$,
 we can rewrite the whole system \eqref{gen12}, \eqref{eq:whole} in the form
\begin{equation}
 \begin{split}
 \dot \omega_1&=\pm\frac{d}{p(d^2+1)}\omega_2^2,\\
\dot \omega_2&=-\dfrac{d}{p}\omega_1\omega_2,\\
 \dot \gamma_1&=-\omega_2\gamma_3,\\
 \dot \gamma_2&=\omega_1\gamma_3,\\
 \dot \gamma_3&=\omega_2\gamma_1-\omega_1\gamma_2,
 \end{split}
\label{eq:sysfin}
\end{equation}
which always has two first integrals
\begin{equation}
 F_1= (d^2+1)\omega_1^2+\omega_2^2, \qquad F_2=\gamma_1^2+\gamma_2^2+\gamma_3^2.
\label{eq:ccalki}
\end{equation}
Note that on the solutions \eqref{omegas} one has $F_1= a^2 c^2= $
As follows from Theorem \ref{pro:necc}, the solution of this system is meromorphic if and only if
$$
p=\sqrt{\dfrac{I_{11}-I_{22}}{I_{11}I_{22}(I_{11}I_{22}I_{33}-1)}} = \frac p d
$$ is a non-zero integer.

Since $d=\sqrt{c^2-1}$ and $p$ can take positive, as well as negative values, we choose
sign $+$ in the first equation in \eqref{eq:sysfin}.

\begin{proposition} For any $p>0$ there exists a ``physical'' rigid body with the inertia tensor
$\mathbb I$ and $I_{13}=0$ such that its eigenvalues satisfy the triangular inequalities and its components satisfy
the meromophisity condition (\ref{eq:necc1}).
\end{proposition}

\begin{proof} In view of the structure of the tensor $\mathbb I$, for its eigenvalues to be positive, the diagonal entries $I_{11}, I_{22}, I_{33}$, as well as $I_{22} I_{33}-I_{23}^2$, must be positive.

Assume that $I_{11}> I_{22}>0$ and $I_{33} + I_{22}> I_{11}$. Then, given such $I_{11}, I_{22}, I_{33}$, the condition (\ref{eq:necc1}) represents a quadratic equation with respect to $I_{23}^2$, whose solutions are
$$
I_{23}^2 = \frac {I_{22}}{2 (I_{11}-I_{22})  } \left( p^2+ 2I_{33}(I_{11}-I_{22})
\pm p  \sqrt{p^2+ 4 I_{33}(I_{11}-I_{22}) } \right).
$$
One can check that both solutions are positive, hence the obtained $I_{23}$ are real.
Let us chose the solution with sign $-$ at the square root. Then we get
$$
I_{22} I_{33}-I_{23}^2 = \frac {I_{22}}{2 (I_{11}-I_{22})  } \left(- p^2+
p \sqrt{p^2+ 2I_{33}(I_{11}-I_{22}) } \right).
$$
Since $p$ and $I_{11}-I_{22}$ are positive, we have $I_{22} I_{33}-I_{23}^2>0$, as required.

Next, the eigenvalues of $\mathbb I$ are $I_{11}, J_2, J_3$, where $J_2, J_3$ are the solutions of
$$
y^2- (I_{22}+I_{33}) y + I_{22} I_{33}-I_{23}^2 =0.
$$
The determinant of the latter equation is $D=(I_{22}- I_{33})^2+4 I_{23}^2>0$, hence both roots are positive.

Finding the roots, one can also check that
$I_{11}^2-(J_2-J_3)^2>0$, which gives the triangular inequalities $I_{11}+J_2> J_3$, $I_{11}+ J_3 >J_2$.
The third inequality follows from the assumption $I_{33} + I_{22}> I_{11}$ and the fact that
$J_2+J_3=I_{33} + I_{22}$.
\end{proof}

\paragraph{Steady-state rotations of the body in space. Comparison with the Chaplygin sleigh.}

Solutions of the Poisson equations (\ref{eq:ps}) with the coefficients $\omega(t)$ in \eqref{eq:o12}
describe the evolution
of the rigid body in space, which is an asymptotic evolution from one steady-state rotation to another one.
The equilibria of $\omega(t)$ correspond to the steady-state rotations themselves.

In the particular case $I_{13}=0$, as $t \to \pm\infty$, the motion in space tends to rotations about the axis $(1,0,0)$
fixed in the body with the angular velocities $\mp a$.
 Note that the {\em spatial orientations} of the axes
of the above steady-state rotations are generally {\em not} the same.

 \begin{figure}[h,t]
 \begin{center}
 \includegraphics[width=.5\textwidth]{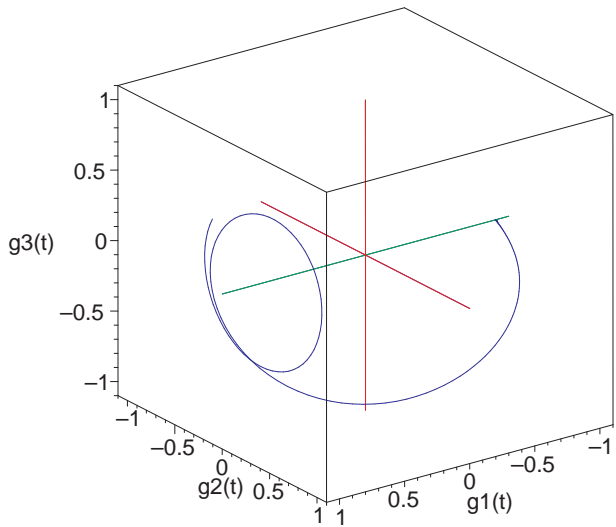}
 \caption{\footnotesize A generic trajectory $\gamma(t)=(g1,g2,g3)^T$ with the
boundary condition (\ref{ss}).}
 \end{center} \label{gen.fig}
 \end{figure}

This can be illustrated by solutions of \eqref{eq:ps} with the initial
(or, rather, boundary) condition
\begin{equation} \label{ss}
\lim_{t\to -\infty} \gamma(t) =(-1,0,0)^T.
\end{equation}
As follows from the structure of \eqref{eq:ps} and \eqref{omegas}, in this case,
as $t\to +\infty$,
$\gamma(t)$ tends to the periodic trajectory
$$
(\cos (\Delta\psi) ,\sin (\Delta\psi) \cos (at),\, \sin(\Delta\psi) \sin (at)
)^T,
$$
where $\Delta\psi$ is the angle between the above
axes of the steady-state rotations in space. That is, the trajectory $\gamma(t)$
in the body frame tends to a uniform rotation about the axis $(1,0,0)$ (see an example in Fig. 2). 
\medskip

A similar behavior occurs in the ``non-compact'' version of the Suslov problem:
the Chaplygin sleigh, a rigid body moving on a horizontal plane supported at three
points, two of which slide freely without friction, while the third is
a knife edge (a blade), which allows no motion orthogonal to its direction
(see for the details, among other, \cite{Ch1911, NeFu}).
The configuration space of this dynamical system is $SE(2)$, the group of
Euclidean motions of the two-dimensional plane ${\mathbb R}^2$,
which can be parameterized by the angular orientation $\theta$ of the blade, and
the position $(x, y)$ of the
contact point of the blade. The presence of the blade determines a nonholonomic
distribution on the phase space $T\, SE(2)$.

Like in the Suslov problem, the equations for the two linear momenta and the
angular momentum in the body frame
separate and are solvable in terms of hyperbolic functions.

In contrast to what is known about the solvability of the Poisson equations \eqref{eq:ps}, the
kinematic equations for the Chaplygin sleigh  are integrable for any initial conditions \cite{NeFu}. In particular,
as $t\to \pm\infty$, the motion of the sleigh tends to a straight line uniform
motion, however along different directions.
(This can be compared with the steady-state rotations of the body about
different axes in the Suslov problem.)
A typical trajectory of the contact point of the blade is shown in Fig. 3. 
In addition, in \cite{NeFu} the following remarkable property was observed: {\it
the angle between the limit straight line motions of the sleigh does not depend on the initial conditions,
but only on dynamical parameters of the sleigh.}

In this connection the following classical questions arise: {\it is it true that
the angle $\Delta\psi$ between the
axes of the limit steady-state rotations of the Suslov problem depends only on
the components of the inertia tensor
$\mathbb I$ (that is, it does not depend on the energy of the motion)} ?  And,
if yes, how $\Delta\psi$ depends on the moments of inertia $I_{ij}$ ?

It appears that the answer is positive and the function $\Delta\psi(I_{ij})$ can be calculated
explicitly without solving the Poisson equations (see Theorem \ref{angle} in Section 5).

\begin{figure}[h,t]
 \begin{center}
 \includegraphics[width=.4\textwidth]{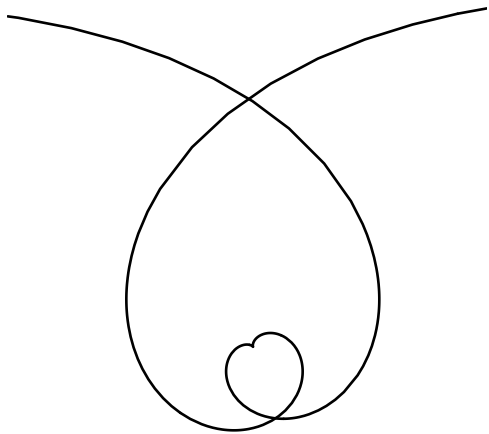}
 \caption{\footnotesize A generic trajectory of the contact point of the blade
on the plane $(x, y)$.}
 \end{center} \label{chapl.fig}
 \end{figure}

\section{Solvability analysis}
\subsection{Solvability in the class of generalized hypergeometric functions}
\label{sec:solhyper}
In this subsection we again assume that $I_{13}=0$ and fix the energy  the level to~\eqref{eq:enlev}.
Hence, $\omega_1(t)$ and $\omega_2(t)$ are given by~\eqref{omegas}, with
\begin{equation}
 a =\frac{p}{d}, \quad d^2=c^2-1.
\end{equation}
Note that here we do not assume that $p\in\Z^\star$.

The Poisson equations \eqref{eq:whole} can be rewritten as the following third order equation
\begin{equation}
 a_0\dddot \gamma_1+a_1\ddot\gamma_1+a_2\dot\gamma_1+a_3\gamma_1=0,
\label{eq:genhipprep}
\end{equation}
where
\begin{equation} \label{a'ss}
 \begin{split}
&a_0 (t) =(d^2+1) p^2\omega_1,\quad a_1 (t) =d p (2 (d^2+1) \omega_1^2 - \omega_2^2),\\
& a_2(t)=(d^2+1)\omega_1[(d^2 + p^2) \omega_1^2 + p^2 \omega_2^2],\quad
 a_3(t) =-d p \omega_2^2 [(d^2+1)\omega_1^2 + \omega_2^2]
 \end{split}
\end{equation}
Now introduce new independent variable
\begin{equation} \label{t->z}
z(t)=\dfrac{d^2}{(d^2+1)p^2}\omega_2^2 \equiv  \frac {4}{(e^t+e^{-t})^2}.
\end{equation}
Then
\[
\begin{split}
& \dot z=-\dfrac{2 d z \omega_1}{p},\qquad \ddot z=\dfrac{2 z (2 p^2 + 2 d^2 p^2
- 3 (1 + d^2) p^2 z)}{(d^2+1) p^2},\\
&\dddot z=-\dfrac{8 d z (p^2 + d^2 p^2 - 3 (d^2+1) p^2 z) \omega_1}{(d^2+1)
p^3}
\end{split}
\]
and \eqref{eq:genhipprep}, \eqref{a'ss} transform to
\begin{equation} \label{eq:gamma_1}
 \gamma_1'''+b_1\gamma_1''+b_2\gamma_1'+b_3\gamma_1=0 \, ,
\end{equation}
where the prime denotes the differentiations with respect to $z$, and
\begin{equation}
 \begin{split}
&b_1=\dfrac{2}{z}+\dfrac{1}{z-1} ,\quad     b_3=\dfrac{(d^2+1) p^2}{8 d^2 (z-1)^2 z^2},              \\
& b_2=\dfrac{-p^2 ( z-1) +
 d^2 (1 + z (-6 - p^2 ( z-1) + 4 z))}{4 d^2 (z-1)^2 z^2} \, .
 \end{split}
\end{equation}
Now, setting $\gamma_1=\sqrt{z-1}u(z)$, we obtain the following equation for the function $u(z)$
\begin{equation}
z^2(1-z)u'''+\dfrac{1}{2} (4 - 9 z) zu''+ \left[\dfrac{1}{4} \left(1 +
\dfrac{p^2}{d^2}\right) + \dfrac{1}{4} ( p^2-13) z\right]u'+\dfrac{1}{8} (
p^2-1)u=0.
\label{eq:genhip}
\end{equation}
This  is  a generalized third order hypergeometric equation, whose canonical form is
(see e.g., \cite{Mimachi_08})
\begin{equation}
\begin{split}
 &z^2(1-z)u'''+z[1+\beta_1+\beta_2-(3+\alpha_1+\alpha_2+\alpha_3)z]u''\\
&+
[\beta_1\beta_2-(1+\alpha_1+\alpha_2+\alpha_3+\alpha_1\alpha_2+\alpha_2\alpha_3+
\alpha_3\alpha_1)z]u'-\alpha_1\alpha_2\alpha_3u=0 \, .
\end{split}
\label{eq:3F2}
\end{equation}
One can identify (\ref{eq:3F2}) and \eqref{eq:genhip} by setting
\begin{equation}
 \begin{split}
\alpha_1
=\dfrac{1}{2},\qquad\alpha_2=\dfrac{1+p}{2},\qquad\alpha_3=\dfrac{1-p}{2},
\qquad\beta_1=\dfrac{d-\rmi p}{2d},\qquad
\beta_2=\dfrac{d+\rmi p}{2d}.
 \end{split}
\label{eq:param}
\end{equation}
Equation \eqref{eq:3F2} belongs to the class of  generalized hypergeometric equations of order $n$ for $n=3$.
If we denote by $\theta=z\mathrm{d}/\mathrm{d} z$, then this class takes the form
\begin{equation}
 [(\theta+\beta_1-1)\cdots
(\theta+\beta_{n-1}-1)\theta-z(\theta+\alpha_1)\cdots(\theta+\alpha_n)]u=0,
\label{eq:nFn-1}
\end{equation}
 see e.g. \cite{Okubo,Beukers}. Recall that one of its solution is
the generalized hypergeometric function introduced by Thomae~\cite{Thomae:1870::}
\begin{equation}
\label{eq:ghyp}
 \phantom{\vert}_nF_{n-1}\left(\begin{matrix}
 \alpha_1, & \ldots,& \alpha_n \\
 \beta_1,& \ldots, & \beta_{n-1}
\end{matrix}; z
  \right):= \sum_{k=0}^{\infty}
\dfrac{(\alpha_1)_k\cdots(\alpha_n)_k z^k}{(\beta_1)_k\cdots(\beta_{n-1})_kk!}.
\end{equation}
Here $(a)_k=a(a+1)\cdots(a+k)$ is the Pochhammer symbol and $(a)_0=1$.

It is known (see e.g., \cite{Mimachi_08}) that equation~\eqref{eq:3F2} has three regular
singularities over $\mathbb{CP}^1$ with the exponents
\begin{eqnarray*}
& 0,1-\beta_1,1-\beta_2,&\mtext{at} z=0,\\
&0,1,\beta_1+\beta_2-\alpha_1-\alpha_2-\alpha_3,&\mtext{at} z=1,\\
& \alpha_1,\alpha_2,\alpha_3,&\mtext{at} z=\infty.
\end{eqnarray*}
and if $\beta_1,\beta_2,\beta_1-\beta_2\not\in\Z$, then the general solution of
\eqref{eq:genhip} around the origin is given by linear combination
\begin{align}
u & = C_1 {\cal F}_1(z) + C_2 {\cal F}_2(z) + C_3{\cal F}_3(z), \label{gen_sol_u} \\
  {\cal F}_1(z) & =  \phantom{\vert}_3F_2\left(\begin{matrix}
 \alpha_1& \alpha_2& \alpha_3 \\
 \beta_1& \beta_2 &
\end{matrix}; z\right)\, , \nonumber \\
 {\cal F}_2(z) &= z^{1-\beta_1}\phantom{\vert}_3F_2\left( \begin{matrix}
\alpha_1-\beta_1+1& \alpha_2-\beta_1+1&
\alpha_3-\beta_1+1\\
\beta_2-\beta_1+1& 2-\beta_1&
\end{matrix}; z\right) \, ,
\label{eq:sol3F2}
\\
 {\cal F}_3(z) & = z^{1-\beta_2}\phantom{\vert}_3F_2
\left( \begin{matrix}
\alpha_1-\beta_2+1& \alpha_2-\beta_2+1&
\alpha_3-\beta_2+1\\
\beta_1-\beta_2+1& 2-\beta_2 &
\end{matrix};z\right) \, , \nonumber
\end{align}
$C_1, C_2$, and $C_3$ being constants of integration.

An immediate observation is that if one of $\alpha_i$ in~\eqref{eq:ghyp}
is a negative integer, then the series truncates. In view of \eqref{t->z},
this happens precisely when $p$ is an {\it odd} integer, and then the particular solution ${\cal F}_1(z)$
takes the form
\begin{equation}
{\cal F}_1(z)=  1+ \sum_{j=1}^{(p-1)/2} \frac{(2j-1)!!}{(2j)!!}\,
\frac {d^{2j} (p^2-1)\cdots (p^2-(2j-1)^2) }{(d^2+p^2)\cdots ((2j-1)d^2+p^2)}  \, z^j \, .
\label{pol_sol}
\end{equation}
Since $\gamma_1=\sqrt{1-z} \, u(z)$ and $z=4/(e^t+e^{-t})^2$,
this corresponds to a solution $\gamma_1(t)$, rational in $e^t$ and meromorphic in $t$, having poles of order $p$
at $t=\pi \rmi/2$ (mod $\pi \rmi$), as predicted by Theorem \ref{pro:necc}.

\paragraph{Transformations for $ {\cal F}_2(z), {\cal F}_3(z)$.}
Due to the symmetry in the definition (\ref{eq:param}) of the parameters $\alpha, \beta$, the above series
$ {\cal F}_2(z), {\cal F}_3(z)$  can be expressed in terms of customary hypergeometric functions
$$
F\left(\alpha,\beta;\gamma; y\right) = 1+ \frac{\alpha \beta}{\gamma} y
+ \frac{\alpha (\alpha+1) \beta(\beta+1)}{\gamma(\gamma+1) 2!} y^2+ \cdots
$$
of the variable $y=e^{2t}$.

\begin{proposition} Under the condition (\ref{eq:param}) one has
\begin{equation}
{\cal F}_{2,3}(z) = \frac{\varkappa_\pm } {(1+y)^{p-1}(y-1)} y^{\dfrac{d\pm\rmi p}{2d}} \cdot  F(a_{\pm},b_{\pm};\, c_{\pm};-y)
\hat F (a_{\pm},b_{\pm};\, c_{\pm};-y), \label{F_23_splitted}
\end{equation}
where
\begin{gather*}
z= 4 y/(1+y)^2, \quad \varkappa_\pm = \dfrac{2d }{(d \pm\rmi p) (3 d \pm\rmi p)}, \\
\hat F =  \left \{ [d(p-1)y-d\mp\rmi p]F(a_{\pm},b_{\pm};c_{\pm}; -y)
-2d\, y(y+1)\dfrac{\mathrm{d}}{\mathrm{d}y}F(a_{\pm},b_{\pm};c_{\pm},-y)\right\},
\end{gather*}
and
\begin{equation} \label{abc}
a_{\pm}=\frac{2-p}{2}\pm\frac{\rmi p}{2d},\quad b_{\pm}=\frac{1-p}{2},\quad
c_{\pm}=1+a_{\pm}-b_{\pm}=\dfrac{3}{2}\pm\dfrac{\rmi p}{2d}.
\end{equation}
\end{proposition}

\begin{proof} In view of (\ref{eq:param}), the last two formal solutions in (\ref{eq:sol3F2})
can be written as
\begin{gather}
{\cal F}_{2,3}(z(y)) = (4y)^{ \frac{d\pm \rmi p }{2d} }(y+1)^{-\frac{d\pm \rmi p }{d} }
\phantom{\vert}_3F_2\left(\begin{matrix}
\alpha_{\pm}+\beta_{\pm} & 2\alpha_{\pm} & 2\beta_{\pm} \\
 2\alpha_{\pm}+2\beta_{\pm}-1& \alpha_{\pm}+\beta_{\pm}+\frac{1}{2} &
\end{matrix}; \frac{4y}{(y+1)^2}\right) , \nonumber \\
\alpha_{\pm}=\dfrac{2-p}{4}\pm\dfrac{\rmi p}{4d},\quad \beta_{\pm}=\dfrac{2+p}{4}\pm\dfrac{\rmi p}{4d}.
\label{alpha,beta}
\end{gather}
The functions $\phantom{\vert}_3F_2(\cdot)$ of such kind split into a product of two customary
hypergeometric functions
$\phantom{\vert}_2F_1(\cdot)\equiv F(\cdot)$, namely
\begin{equation}
\phantom{\vert}_3F_2\left(\begin{matrix}
 2\alpha& 2\beta& \alpha+\beta\\
 2\alpha+2\beta-1& \alpha+\beta+\frac{1}{2} &
\end{matrix}; x\right)= F\left(\alpha,\beta;\gamma;x\right) \, F\left(\alpha,\beta;\gamma-1;x\right) ,
\label{splits}
\end{equation}
where $\gamma=\alpha+\beta+ \frac{1}{2}$, see e.g. formula (8) on page 86 in  \cite{Erdelyi:81::a}.
In our case
$$
\alpha=\alpha_{\pm}, \quad \beta=\beta_{\pm}, \quad
\gamma_{\pm}=\alpha_{\pm}+\beta_{\pm}+ \frac{1}{2}=\dfrac{3}{2}\pm\dfrac{\rmi p}{2d}, \quad
x=4y/(y+1)^2.
$$
Next, due to the special form of $\gamma$,
the first hypergeometric function $F\left(\alpha,\beta;\gamma;x\right)$ can be written as series in $y$ by using
the following quadratic transformation
\begin{equation}
 F\left(\frac{a}{2},\frac{a}{2}+\frac{1}{2}-b,1+a-b, \frac{4y}{(1+y)^2}\right)=(1+y)^{a}F(a,b,1+a-b,-y),
\label{eq:fferd}
\end{equation}
see e.g. page 64 in \cite{Erdelyi:81::a}.
Now identifying $\alpha_{\pm}=a_{\pm}/2$ and $\beta_{\pm}=a_{\pm}/2+1/2-b_{\pm}$, we obtain
$$
 a_{\pm}=\frac{2-p}{2}\pm\frac{\rmi p}{2d},\quad b_{\pm}=\frac{1-p}{2},\quad
c_{\pm}=1+a_{\pm}-b_{\pm}=\dfrac{3}{2}\pm\dfrac{\rmi p}{2d},
$$
which coincide with \eqref{abc}. Inserting these parameters into \eqref{eq:fferd}, we get
\begin{equation}
 F\left(\alpha_{\pm},\beta_{\pm};\gamma_{\pm};\frac{4y}{(y+1)^2}\right)=(y+1)^{a_{\pm}}F(a_{\pm},b_{\pm};c_{\pm},-y).
\label{eq:expsol}
\end{equation}
It remains to apply a similar quadratic transformation to the factor
$F(\alpha_{\pm},\beta_{\pm};\gamma_{\pm}-1;4y/(y+1)^2)$. To do this we first use the relation
\[
 (\gamma-n)_n\, z^{\gamma-1-n}F(\alpha,\beta;\gamma-n;z)
=\dfrac{\mathrm{d}^n}{\mathrm{d}z^n}\left[z^{\gamma-1}F(\alpha,\beta;\gamma;z)\right],
\]
$(\gamma-n)_n$ being the corresponding Pochhammer symbol (see e.g. equation (22) on page 102 in  \cite{Erdelyi:81::a}).
Setting here $n=1$, we obtain
\begin{equation}
\begin{split}
 &F(\alpha_{\pm},\beta_{\pm};\gamma_{\pm}-1;z)
=\dfrac{1}{\gamma_{\pm}}F(\alpha_{\pm},\beta_{\pm};\gamma_{\pm};z)+\dfrac{z}{\gamma_{\pm}(\gamma_{\pm}-1)}
\dfrac{\mathrm{d}}{\mathrm{d}z}F(\alpha_{\pm},\beta_{\pm};\gamma_{\pm};z),\\
&z=\frac{4y}{(y+1)^2}.
\end{split}
\label{eq:expsol1}
\end{equation}
In view of the chain rule
\[
 \dfrac{\mathrm{d}}{\mathrm{d}z}=-\dfrac{(y+1)^3}{4(y-1)}\dfrac{\mathrm{d}}{\mathrm{d}y},
\]
and \eqref{alpha,beta}, the expression \eqref{eq:expsol1} takes the form
\begin{align*}
& F \left( \alpha_{\pm},  \beta_{\pm};\gamma_{\pm}-1; \frac{4y}{ (y+1)^2} \right) =
\dfrac{2d(y+1)^{a_{\pm}}}{(d \pm\rmi p) (3 d \pm\rmi p)(y-1)} \\
& \quad \cdot \left\{
[d(p-1)y-d\mp\rmi p]F(a_{\pm},b_{\pm};c_{\pm},-y)-2dy(y+1)\dfrac{\mathrm{d}}
{\mathrm{d}y}F(a_{\pm},b_{\pm};c_{\pm},-y)\right\}.
\end{align*}

Substituting this, as well as \eqref{eq:expsol}, into \eqref{splits} and simplifying, we finally get the expression
\eqref{F_23_splitted}. \end{proof}

\medskip

Note that when $p$ is odd positive integer, the parameters $ b_{\pm}=(1-p)/2$ become negative integer,
and the hypergeometric series $ F (a_{\pm},b_{\pm};c_{\pm},-y) $ in (\ref{F_23_splitted})
convert into polynomials of degree $(p-1)/2$,
\begin{equation} \label{poly_F}
F (a_{\pm},b_{\pm};c_{\pm},-y) = \sum_{k=0}^{(p-1)/2} \frac{(a_\pm)_{k} (b_\pm)_{k}}{(c_\pm)_{k} k!}\, (-y)^k .
\end{equation}
Next, after a simplification,
the series $\hat F (a_{\pm},b_{\pm};c_{\pm},-y)$ in (\ref{F_23_splitted}) becomes the following polynomial of degree $(p-1)/2$
(the term of the highest degree $(p-1)/2+1$ annihilates):
\begin{equation}
\hat F (a_{\pm},b_{\pm};c_{\pm},-y) =
\sum_{k=0}^{(p-1)/2} \frac{(a_\pm)_{k} (b_\pm)_{k}}{(c_\pm)_{k} k!}
\, \frac{(2d \, k + d+ p\rmi )(4d \, k -dp \pm \rmi p)  }{dp-2d\, k \mp \rmi p } (-y)^k . \label{poly_F2}
 \end{equation}
As result, we conclude that in this case,
\begin{equation}
\label{F_23_pol}
{\cal F}_{2,3} (z(y)) = \frac {\varkappa_\pm }{ (1+y)^{p-1}(y-1)}\, y^{ \frac{d\pm\rmi p}{2d} }\, P_\pm (y),
\end{equation}
where $P_\pm (y)= F (a_{\pm},b_{\pm};c_{\pm},-y)\, \hat F (a_{\pm},b_{\pm};c_{\pm},-y)$
are polynomials of degree $p-1$ with complex conjugated coefficients.

The fact that in this case the series ${\cal F}_1$ is a rational function of $y$
follows from the formula \eqref{pol_sol}. 


The explicit form of the corresponding meromorphic solutions of the Poisson equations \eqref{eq:whole} for odd $p$
is given in Section~\ref{mer_sol}.

\paragraph{Monodromy for ${\cal F}_1(z)$ and the angle $\Delta \psi$.}
Now we show how using the monodromy group of the equation~\eqref{eq:nFn-1},
which was analized in \cite{Okubo, Beukers, Mimachi_08}, one can calculate the angle $\Delta \psi$ between the
axes of the limit steady-state rotations of the body in space for generic $p\in {\mathbb R}$.
As follows from \eqref{t->z}, the values
$t=\pm \infty$, $t=0$, and $t=\pi i/2$ (mod $\pi i$) correspond to the singular points
$z=0$, $z=1$, and $z=\infty$ respectively. As time $t$ evolves along the real axis from $-\infty$ to $\infty$ passing by
$t=0$, the variable $z$ makes a loop in $\mathbb C$ embracing $1$ in positive or negative direction.

Next, we need
\begin{lemma}
Under the substitution \eqref{t->z},  the component $\gamma_1$ of
the special solution $\gamma(t)\in S^2$ satisfying the boundary condition \eqref{ss} has the following form
around the origin $z=0$
\begin{align}
\gamma_1(t) & = - \sqrt{1-z(t)} \, {\cal F}_1 (z(t)) \nonumber \\
         & \equiv - \sqrt{1-z} \phantom{\vert}_3F_2\left(\begin{matrix}
 1/2 &  (1+p)/2   & (1-p)/2  \\
 1/2-ip/(2d)  &  1/2 + ip/(2d)  &
\end{matrix}; z\right) \, . \label{gam_0}
\end{align}
\end{lemma}

\begin{proof} As follows from \eqref{gen_sol_u},
the most general possible expression for the component $\gamma_1$, which ensures
$\lim_{t\to -\infty} \gamma_1=-1$, is 
$$
 \widehat \gamma_1 (t) =  - \sqrt{1-z} \,\left( {\cal F}_1 (z) + c_2 {\cal F}_2(z)+ c_3 {\cal F}_3(z) \right),
$$
$c_2, c_3$ being arbitrary constants. On the other hand, from the Poisson equations and the solutions (\ref{omegas})
one has
$$
\gamma_3 =- \frac {\dot \gamma_1}{\omega_1} = -4 \frac{ e^t-e^{-t} }{ac(e^t + e^{-t})^2 } \frac{d\gamma_1}{d z}
\equiv - \frac{e^t-e^{-t}}{ac} z \frac{d\gamma_1}{d z}.
$$
The above special solution requires that $\lim_{t\to -\infty} \gamma_3=0$.
Then, setting here $\gamma_1= \widehat \gamma_1$, differentiating ${\cal F}_{1,2,3}(z)$ and taking into account the values of $\beta_1,\beta_2$, we see that this condition holds if and only if $c_2, c_3$ above are zero.
\end{proof}

When $z$ makes a loop around $1$, the root $\sqrt{1-z}$ changes sign, whereas
the hypergeometric series ${\cal F}_1 (z)$ transforms according to the following proposition,
which is a corollary of Theorem 4.1 in \cite{Mimachi_08}.

\begin{proposition} When $z\in {\mathbb C}$ makes a loop around the singular point $z=1$ in positive direction, the
solution ${\cal F}_1(z)$ around the origin undergoes the monodromy
\begin{align}
{\cal F}_1 \to {\cal F}_1 & -2 \rmi  \exp(\beta_1+\beta_2-\alpha_1-\alpha_2-\alpha_3)
\left( \sigma_1 {\cal F}_1 + \sigma_2 {\cal F}_2 + \sigma_3 {\cal F}_3 \right)  \\
 \sigma_1 & = \frac{ \sin(\pi \alpha_1)\, \sin(\pi \alpha_2) \, \sin(\pi \alpha_3)} {\sin(\pi \beta_1)\, \sin(\pi \beta_2) },
\nonumber \\
\sigma_2 & = -\frac{ \sin(\pi ( \beta_1-\alpha_1) )\, \sin(\pi (\beta_1-\alpha_2))\, \sin(\pi (\beta_1-\alpha_3)}
{\sin(\pi \beta_1) \sin(\pi (\beta_1-\beta_2)) } , \nonumber \\
\sigma_3 & = -\frac{ \sin(\pi ( \beta_2-\alpha_1) ) \, \sin(\pi (\beta_2-\alpha_2))\, \sin(\pi (\beta_2-\alpha_3)}
{\sin(\pi (\beta_2-\beta_1))\, \sin(\pi \beta_2) } . \nonumber
\end{align}
\end{proposition}

Now combining formula (\ref{gam_0}) with the proposition, taking into account \eqref{eq:param}, then
evaluating $ {\cal F}_1, {\cal F}_2, {\cal F}_3$ at $z=0$, one concludes that
$$
\lim_{t\to \infty} \gamma_1 = 1-2 \rmi  \exp(\pi \rmi /2) \sigma_1
= 1- 2 \frac{ \cos^2(\pi p/2) }{\cos^2( (\pi \rmi p)/(2q) ) } .
$$
Since $\lim_{t\to -\infty} \gamma_1=-1$ and $\cos(\Delta \psi) = - \lim_{t\to \infty} \gamma_1$, we obtain the following

\begin{theorem} \label{angle}
The angle $ \Delta \psi $ between the axes of the limit steady-state rotations of the body does not depend on the energy
and, when $I_{13}=0$, is uniquely defined by relation
$$
 \cos \frac{\Delta \psi }2= \frac{ \cos(\pi p/2) }{\cos( (\pi \rmi p)/(2q) )} \, ,
$$
where, as above,
$$
d= \sqrt{\frac {I_{11}-I_{22}}{I_{22}} }, \quad p= \sqrt{ \dfrac{I_{11}-I_{22}}{I_{11}I_{22}(I_{11}I_{22}I_{33}-1)} }\,, 
$$
and it is assumed that $\det {\mathbb I}=1$.
\end{theorem}

Note that when the parameter $p$ is odd integer, $\Delta \psi$ is always $\pi$, regardless to value of $q$.

We also add that the formula of Theorem \ref{angle} stands in a perfect correspondence with numerical integration
tests.

It remains to study the case of {\it even} integer $p$, when, according to Theorem \ref{pro:necc}, all the solutions
of the Poisson equations \eqref{eq:whole} are meromorphic,
but are not given by truncated generalized hypergeometric series.
This will be one of the subjects of next subsection.

\subsection{Differential Galois analysis}
Here we recall  one result of G.~Darboux  which  was formulated in Chapter II of
his \emph{Th\'eorie g\'en\'erale des surfaces}.
\begin{lemma}
\label{lem:dar}
Assume that $\gamma^{(1)}= \gamma^{(1)}(t)$ is a real solution of the Poisson
equation
\begin{equation}
 \label{eq:ryba}
\dot \gamma = \gamma\times\omega(t),
\end{equation}
where $\omega(t)$ is a real vector, and $\gamma^{(1)}$ satisfies
$\langle\gamma^{(1)},\gamma^{(1)}\rangle=1$.  Then the remaining two solutions
of~\eqref{eq:ryba} linearly independent with $\gamma^{(1)}$  can be found with
by a single quadrature.
\end{lemma}
\begin{proof}
 Let us restrict equation~\eqref{eq:ryba} to the unit sphere
$\langle\gamma,\gamma\rangle=1$. As coordinates on it we choose
\begin{equation}
\label{eq:ui}
 u_1=\frac{\gamma_3+1}{\gamma_1-\rmi \gamma_2}, \quad
u_2=-\frac{\gamma_1+\rmi \gamma_2}{\gamma_3+1},
\end{equation}
so
\begin{equation}
\label{eq:gi}
 \gamma_1 = \frac{1-u_1u_2}{u_1-u_2}, \quad  \gamma_2 =
\rmi\frac{1+u_1u_2}{u_1-u_2}, \quad  \gamma_3 = \frac{u_1+u_2}{u_1-u_2}.
\end{equation}
Now, it is easy to check that
 $u_1$ and $u_2$  satisfy  the following Riccati equation
\begin{equation}
\label{eq:rig}
 \dot u=A +Bu +Cu^2,
\end{equation}
where
\begin{equation}
 \label{eq:ABC}
A=\frac{1}{2}(\omega_2(t)-\rmi\omega_1(t)), \quad B=-\rmi \omega_3(t), \quad
C=\frac{1}{2}(\omega_2(t)+\rmi\omega_1(t))
\end{equation}
Formulae~\eqref{eq:gi} show that knowing two different solutions
of~\eqref{eq:rig} we determine one solution of~\eqref{eq:ryba}. On the other
hand, having one real solution $\gamma^{(1)}$ of~\eqref{eq:ryba} we have two
different solutions of Ricatti equation~\eqref{eq:rig} which are given by
formulae~\eqref{eq:ui}.

Let  $u_0$ be a solution of~\eqref{eq:rig}. Then, as it is easy to check
$u_1=-1/u_0^\star$, where $z^\star$ denotes the complex conjugate of $z$, is
also a solution of this equation. It is well known that the general solution $u$
of a Ricatti is determined by its three different particular solutions $u_0$,
$u_1$ and $u_2$. Namely, we have
\begin{equation}
 \frac{(u-u_0)}{(u-u_1)}\frac{(u_2-u_1)}{(u_2-u_0)} = C_0
\end{equation}
where $C_0$ is an arbitrary constant. From this formula we deduce that knowing
only two different solutions $u_0$ and $u_1$, the general solution can be obtain
by a single quadrature
\begin{equation}
 \frac{u-u_0}{u-u_1}=C_0\exp\left[\int^t C(s)(u_0(s)-u_1(s))\rmd s\right].
\end{equation}
\end{proof}

Putting $u=-v/C$   transform equation~\eqref{eq:rig} to the form
\begin{equation}
\label{eq:r2}
 \dot v = A_1 +B_1 v -v^2, \mtext{where} A_1=-AC, \quad B_1= B+ \frac{\dot
C}{C}.
\end{equation}
Then $v=\dot w/w$ where
\begin{equation}
\label{eq:lin}
 \ddot w -B_1 \dot w -A_1w=0.
\end{equation}

The general question is whether we can find an explicit form of solutions of
Poisson equation~\eqref{eq:ryba} for given functions $\omega_i(t)$. A proper
setting to this question is given by the differential algebra. Namely, we assume
that
$\omega_i(t)$ are elements of certain differential field $L$  with $\C$ as a
subfield of constants. In the considered
Suslov problem $L=\C(\rme^t)$.

Let $K \supset L$ be the Picard--Vessiot extension for equation~\eqref{eq:ryba}.
We say that the equation is solvable iff extension $K\supset L$ is a Liouvillian
extension. In this case all solution of the equation are Liouvillian.  By the
known Kolchin theorem, the extension $K\supset L$ is a Liouvillian extension iff
the identity component of the differential Galois group $G(K/L)$ is solvable.

\begin{proposition}
\label{prop:1}
 Assume that equation~\eqref{eq:ryba} has a solution  Liouvillian over
$L=\C(\rme^t)$. Then all its solutions are Liouvillian.
\end{proposition}
\begin{proof}
Let $\gamma^{(1)}$ be a Liouvillian solution of~\eqref{eq:ryba}. Assume that
$\langle \gamma^{(1)}, \gamma^{(1)}\rangle=1$.   From the proof of
Lemma~\ref{lem:dar} we know that  to find other solutions of~\eqref{eq:ryba} it
is enough to find a general solution of the Ricatti equation~\eqref{eq:rig}. But
 $\gamma^{(1)}$  gives us one Liouvillian solution $u_0$ of
equation~\eqref{eq:rig}. The general solution of~\eqref{eq:rig} is given by
\begin{equation}
 u= u_0 + \frac{1}{y},
\end{equation}
where $ y$ satisfies linear equation
\begin{equation}
 \dot y = -(B+2Cu_0)y-C.
\end{equation}
Hence $y$ is Liouvillian and all solution of ~\eqref{eq:rig} are Liouvillian,
and thus all solutions of~\eqref{eq:ryba} are Liouvillian.

If $\gamma^{(1)}$ is a Liouvillian solution of~\eqref{eq:ryba} and $\langle
\gamma^{(1)}, \gamma^{(1)}\rangle=0$, then, as $\omega_i(t)$ are real,  $\Re
\gamma^{(1)}$ and $\Im  \gamma^{(1)}$ are linearly independent real solutions
of~\eqref{eq:ryba}.  Moreover
\begin{equation}
 \gamma= \alpha_1 \Re \gamma^{(1)} + \alpha_2 \Im  \gamma^{(1)} + \alpha_3 ( \Re
\gamma^{(1)})\times (\Im  \gamma^{(1)})
\end{equation}
is a solution of~\eqref{eq:ryba} for arbitrary $\alpha_1$, $\alpha_2$ and
$\alpha_3$. In effect all the solutions off~\eqref{eq:ryba} are Liouvillian.
\end{proof}

From the above proved fact we obtain the following.
\begin{proposition}
 Poisson equation~\eqref{eq:ryba} has  a Liouvillian solution iff all solutions
of linear equation~\eqref{eq:lin} are Liouvillian.
\end{proposition}
\begin{proof}
If equation~\eqref{eq:ryba} has a Liouvillian solution, then, by
Proposition~\ref{prop:1}, all its solutions are Liouvillian and
formulae~\eqref{eq:ui} show that all solutions of the Riccati
equation~\eqref{eq:rig}, as well as of the transformed
Riccati equation~\eqref{eq:r2}, are Liouvillian.  As $v=\dot w/w$, $w$ is
Liouvillian, and thus all solutions of~\eqref{eq:lin} are Liouvillian.

On the other hand, if equation~\eqref{eq:lin} has a non-zero Liouvillian
solution, then all solutions of Riccati equation~\eqref{eq:r2}, as well as
of~\eqref{eq:rig} are Liouvillian. Thus, by formulae~\eqref{eq:gi}, Poisson
equation~\eqref{eq:ryba} has a Liouvillian solution, so, by
Proposition~\ref{prop:1}, all its solutions are Liouvillian.
\end{proof}

For the considered problem $\dot{ \gamma}=\gamma\times\omega(t)$,
equation~\eqref{eq:lin} reads
\begin{equation}
 w'' +p(z)w' +q(z)w=0,\quad z=\rme^t,\quad '\equiv \frac{\rmd \phantom{z}}{\rmd
z},
\label{eq:2rzedu}
\end{equation}
where
\begin{equation}
\begin{split}
 p(z)&:= \frac{z^2(z^2+4\rmi c z -4)-1}{z(z^2+1)(z^2+2\rmi c z -1)}, \\
q(z)&:=\frac{p^2}{4(c^2-1)}\frac{1}{z^2}+\frac{p^2}{(z^2+1)^2}.
\end{split}
\end{equation}
This equation is derived under assumption that relation~\eqref{eq:cpr} holds
true. In what follows we will work with the reduced form of the above equation,
i.e., with
\begin{equation}
\label{eq:rlin}
 y''=r(z)y \mtext{where} r(z)
=\frac{P}{Q}=\frac{1}{2}p'(z)+\frac{1}{4}p(z)^2-q(z).
\end{equation}
where $P$ is a polynomial of eight degree
\begin{equation}
 P=\sum_{i=0}^8p_i z^i,
\end{equation}
with the following coefficients
\begin{gather}
 p_0=p_8=d^2+p^2, \quad p_1^\star=p_7=4 \rmi c \left(2 d^2+p^2\right), \quad
p_2=p_6=-4 \left(4 d^2+p^2\right),\\
p_3^\star=p_5=4 \rmi c \left(\left(4
   p^2-2\right) d^2+p^2\right), \\
p_4=-2 \left(\left(-8 d^2+8
   \left(d^2+2\right)
   p^2+1\right) d^2+5
   p^2\right), \quad d^2=c^2-1,
\end{gather}
and
\begin{equation}
\label{eq:Q}
Q=4 \left(c^2-1\right) z^2
   \left(z^2+1\right)^2
   \left(2 c z-\rmi
   \left(z^2-1\right)\right)^2.
\end{equation}
Note, that we assumed that $c\neq\pm 1$, i.e., taking into account
relation~\eqref{eq:cpr}, $p\neq0$. This case we consider later.

From~\eqref{eq:Q} it follows that equation~\eqref{eq:rlin} has six regular
singularities $s_i$, $i=0,\ldots 5$, and $s_0=0$, $s_1=s_2^\star=\rmi$,
$s_3=-\rmi(c+d)$, $s_4=-\rmi(c-d)$, and $s_5=\infty$. The respective differences
of exponents $\Delta_i$ at these points are following
\begin{equation}
 \label{eq:Deltai}
\Delta_0=\Delta_5=\rmi\frac{p}{d},\quad \Delta_1=\Delta_2=p, \quad
\Delta_3=\Delta_4=2.
\end{equation}
We observe immediately that at four singularities $s_1$, $s_2$, $s_3$ and $s_4$
differences of exponents are integer thus at local expansions of solutions
around of this points logarithmic terms can appear. Here we only sketch how this
happens, for more details see e.g.~\cite{Whittaker:35::}.
For linear equations of the second order local solutions $y_1$ and $y_2$ around
a singular point $s_i$ are postulate in the form infinite series with leading
term that are equal to the exponents at this point.
\[
 \rho_i^{(1/2)}=\dfrac{1}{2}(1\pm\Delta_i).
\]
In the case when differences of exponents $\Delta_i=\rho_i^{(1)}-\rho_i^{(2)}$
is integer such two expansions are in general not functionally  independent and
then $y_2$  is constructed in different way from $y_1$ by a quadrature that
usually leads to logarithms.

By direct calculations one can check that solutions around  $s_3$ and $s_4$ do
not have such terms independently on value of $p$. For singular points $s_1$ and
$s_2$ direct calculations for small $p$ show that the following conjecture is
true
\begin{conjecture}
Singular points $s_1$ and $s_2$ are non-logarithmic for any odd integer and
logarithmic for any even $p\in\Z^{\ast}$.
\end{conjecture}
 For singular points $s_1$ and $s_2$ the difference of exponents depends on $p$
and when $p$ grows we have calculate the expansions of $w_1$ with more and more
terms. By this reason we cannot to check the presence of logarithmic terms for
an arbitrary $p$ but we can make such calculations effectively for any chosen
value of $p$ and we did this up to $p=10$.

The presence of logarithmic terms restricts very strongly the possible forms of
elements from the differential Galois group, thus at first up to end of this
section we will consider the case $p$ even and we prove.
\begin{proposition}
\label{pro:a}
 The differential Galois group of equation~\eqref{eq:rlin} for even
$p\in\Z^{\ast}$  is $\mathrm{SL}(2,\C)$.
\end{proposition}
\begin{proof}
The presence of logarithms means that  the differential Galois group of the
considered equation cannot be neither a subgroup of the infinite dihedral group
(because it contains a non-diagonalizable element), nor a finite group. If it is
contained in triangular group, then by the same reason it cannot be its proper
subgroup i.e. diagonal subgroup.
Thus we have two possibilities that it is contained
\begin{enumerate}
 \item in the whole triangular group,
\item is $\mathrm{SL}(2,\C)$.
\end{enumerate}

 If the first possibility  occurs, then equation~\eqref{eq:rlin} has an
exponential solution of the form
\begin{equation}
 y=R(z)\prod_{i=0}^4(z-s_i)^{\rho_i}, \mtext{where} R(z)\in\C[z],
\label{eq:expono}
\end{equation}
and $\rho_i\in\{\rho_i^{(1)},\rho_i^{(2)}\}$ are exponents at points $s_i$, for
$i=0,\ldots,4$. Expanding this solution at infinity we find that
\begin{equation}
 n=-\sum_{i=0}^5\rho_i\geq 0, \mtext{where} n=\deg R.
\end{equation}
We have
\begin{equation}
 \rho_i=\frac{1}{2}\left(1\pm \Delta_i\right) \mtext{for}i=0,\ldots, 4,
\end{equation}
and
\begin{equation}
 \rho_5= -\frac{1}{2}\left(1\pm\Delta_5\right).
\end{equation}
The above implies that we have to choose $\rho_0$ and $\rho_5$ such that $\rho_0
+ \rho_5=0$. As point $s_1$ and $s_2$ are logarithmic, then the only choice is
$\rho_1^{(1)}=\rho_2^{(1)}=(1+p)/2$ for $p>0$, or
$\rho_1^{(2)}=\rho_2^{(2)}=(1-p)/2$ for $p<0$. At $s_3$ and $s_4$ both exponents
$\rho_i^{(1)}=3/2$ and $\rho_i^{(2)}=-1/2$ for $i=3,4$, are possible. Thus for
$p>0$, $n\in\{-p,-p-2,-p-4\}$ and for $p<0$
we have $n\in\{p,p-2,p-4\}$ but all these admissible values are negative. This
shows that the equation does not have an exponential solution \eqref{eq:expono}.
Thus,  the differential Galois group of the equation is $\mathrm{SL}(2,\C)$.
\end{proof}

By the Kolchin theorem this means that in this case the Poisson equations are not
solvable or, more precisely, the following theorem holds.

\begin{theorem}
 Euler-Poisson equations in the meromorphic case defined by~\eqref{eq:necc2} for
even $p\in\Z^{\ast}$ are not solvable in the class of Liouvillian functions.
\end{theorem}

\section{Explicit meromorphic solutions and first integrals} \label{mer_sol}
\subsection{Meromorphic solutions}

As was observed in subsection~\ref{sec:solhyper},
for odd $p\in\N$ the third order hypergeometric equation~\eqref{eq:genhip} for the variable $u=\gamma_1/\sqrt{z-1}$
has three independent quasi-polynomial solutions ${\cal F}_{1,2,3}(z)$ given by \eqref{pol_sol} and \eqref{F_23_pol}, that is,
\begin{gather*}
{\cal F}_{1}(z) = Q(z), \quad
{\cal F}_{2,3} (z(y)) = \frac {\varkappa_\pm}{ (1+y)^{p-1} (y-1)}\, y^{ \frac{d\pm\rmi p}{2d} }\, P_{\pm}(y), \\
z=\frac{4}{(e^t+e^{-t})^2}, \quad y=e^{2t},
\end{gather*}
where $Q(z)$, $P_{\pm}(y)$ are polynomials of degree $(p-1)/2$ and $p-1$ respectively.

Now, taking into account the relation between the solutions $w(z)$ of this hypergeometric equation
and those of the Poisson equation (\ref{eq:genhipprep}), we arrive at

\begin{theorem} \label{merom_gam_1}
For odd $p\in\N$ the third order equation \eqref{eq:genhipprep} for $\gamma_1(t)$ has independent
meromorphic solutions
\begin{align} \label{gam1}
\gamma_1^{(1)} (t) &= \frac {e^t-e^{-t}}{(e^t+ e^{-t})^p} \sum_{k=0}^{(p-1)/2} a_k (e^t+e^{-t})^{p-1-2k}, \\
& \quad a_k = 4^k \, \frac{(2k-1)!!}{(2k)!!}\,
\frac {d^{2k} (p^2-1)\cdots (p^2-(2k-1)^2) }{(d^2+p^2)\cdots ((2k-1)d^2+p^2)} \, , \quad a_0=1, \nonumber \\
\gamma_1^{(2)}(t) &= \frac {e^{ (1-\frac{\rmi p}{d})t } } {(1+ e^{2t})^p} \sum_{k=0}^{(p-1)/2} b_k e^{2k t}
\equiv e^{-\frac{\rmi p}{d} t }\, \frac {b_0 e^{-(p-1)t}+ \cdots + b_{p-1} e^{(p-1)t}}{(e^t+ e^{-t})^p}  ,
\label{gam2} \\
\gamma_1^{(3)}(t) & =\frac {e^{(1+\frac{\rmi p}{d})t}}{(1+ e^{2t})^p} \sum_{k=0}^{(p-1)/2} b_k^* e^{2k t}
\equiv e^{\frac{\rmi p}{d} t }\, \frac {b_0^* e^{-(p-1)t}+ \cdots + b_{p-1}^* e^{(p-1)t}}{(e^t+ e^{-t})^p} ,
\label{gam3}
\end{align}
where $(\cdot)^{\ast}$ denotes the complex conjugation, and the coefficients $b_i$ are uniquely determined from
the polynomial product
\begin{gather*}
\sum_{k=0}^{p-1} b_k y^{k}  =  \sum_{k=0}^{(p-1)/2} \frac{(a)_{k} (b)_{k}}{(c)_{k} k!}\, (-y)^k \cdot
\sum_{k=0}^{(p-1)/2} \frac{(a)_{k} (b)_{k}}{(c)_{k} k!}
\, \frac{(2d k + d+ p\rmi )(4d k -dp + \rmi p)  }{dp-2dk - \rmi p } (-y)^k  , \\
a =\frac{2-p}{2}+\frac{\rmi p}{2d}, \quad b =\frac{1-p}{2},\quad c= \dfrac{3}{2}+ \dfrac{\rmi p}{2d}.
\end{gather*}
\end{theorem}

The proof is straightforward: Substituting the expression \eqref{pol_sol} into
$$
\gamma = \sqrt{z-1} {\cal F}_1(z) = \frac{ e^t-e^{-t} }{e^t + e^{-t} } \,
{\cal F}_1\left(\frac{ 4 }{(e^t + e^{-t})^2 }  \right)
$$
yields \eqref{gam1}. Next, in view of \eqref{F_23_pol}, the product
$$
\gamma = \frac{ e^t-e^{-t} }{e^t + e^{-t} } \, {\cal F}_{2,3} (z(y))
$$
gives \eqref{gam2}, \eqref{gam3}. $\square$
\medskip

Note that all these solutions satisfy
\begin{equation} \label{asymp}
\gamma_1^{(1)}(-t) = \gamma_1^{(1)}(t), \quad
\lim_{t\to \pm\infty} \gamma_1^{(1)}(t) =\mp 1, \quad \lim_{t\to \pm\infty}\gamma_1^{(2,3)}= 0.
\end{equation}

Now let $\vgamma^{(i)} =( \gamma^{(i)}_1 , \gamma^{(i)}_2, \gamma^{(i)}_3)^T$ be the corresponding
vector solutions of the Poisson equations \eqref{eq:ps}. Given $\gamma^{(i)}_1$ as in Theorem \ref{merom_gam_1},
the two remaining components of $\vgamma^{(i)}$, can be calculated by differentiations, using the system \eqref{eq:ps}.

Namely, let us write the above meromorphic solutions in form
$$
\gamma_1^{(1)} = \frac{P_1 (x)}{(1+x^2)^p}, \quad
\gamma^{(2)}_1 = e^{- \frac{\rmi p}{d} t }\, \frac{x P_{2}(x) }{(1+x^2)^p},
\qquad x=e^t ,
$$
$P_1(x), P_{2} (x)$ being polynomials of degree $2p$ and $2(p-1)$ respectively.
Then, in view of \eqref{eq:ps}, one has
\begin{equation}
\gamma_3= -\dfrac{ x \gamma_1'}{\omega_2}, \quad \text{and then} \quad
\gamma_2= \dfrac{\omega_2 \gamma_1- x \gamma_3'}{\omega_1}, \qquad ()'=\frac{d}{dx} .
\label{eq:gam3}
\end{equation}
Applying these formulas and using the expression (\ref{omegas}) for $\omega_1, \omega_2$, we get
\[
 \begin{split}
& \gamma_3^{(1)}= \dfrac{d ( (1 + x^2)P_1'(x)-2 p x P_1(x))}{2p\sqrt{d^2+1}(1 + x^2)^p},\\
& \gamma_2^{(1)}= -\dfrac{x[d^2 (x^2+1) ((x^2+1)P_1''(x)+(2 - 4 p) x P_1'(x))+2 p (2 p
+ d^2 (2 p-1) (x^2+1)) P_1(x)]}{2p^2\sqrt{d^2+1}(1 + x^2)^p(x^2-1)}.
\end{split}
\]
and
\[
 \begin{split}
& \gamma_3^{(2)}= \dfrac{x^{-\frac{\rmi p}{d}}[d x ( x^2+1)P_2'+(d - \rmi p + (d
- \rmi p - 2 d p) x^2) P_2]}{2p\sqrt{d^2+1}(1 + x^2)^p},\\
& \gamma_2^{(2)}= \dfrac{x^{-\frac{\rmi p}{d}}W}{2p^2\sqrt{d^2+1}(1 +
x^2)^p(x^2-1)},\\
&W=-d x (x^2+1)(d x (x^2+1)P_2''(x)+2 (d - 2 d ( p-1) x^2 - \rmi p (x^2+1))P_2'(x))\\
&+(p^2 (x^2-1)^2 - 2 d^2 (1 + p (2 p-3)) x^2 (x^2+1) -
   \rmi d p (x^2+1) (-1 + (4 p-3) x^2)) P_2(x)
\end{split}
\]
Expressions for $\gamma_i^{(3)}$ are complex conjugations of those for $\gamma_i^{(2)}$.

It remains to mention that all the components of $\vgamma^{(i)}$ have poles of order $p$ at
$t=\frac{\pi}{2}\rmi + \pi\rmi N$, $N\in {\mathbb Z}$.


\paragraph{Remark.} Note that for real moments of inertia $I_{ij}$ and, therefore, real constant $d$,
the solution \eqref{gam1} is real. To get real solutions for the other components one just takes real and imaginary
parts of \eqref{gam2}, \eqref{gam3}, as well as $\gamma^{(i)}_2, \gamma^{(i)}_3$,
using the formula $e^{\rmi \alpha t}= \cos(\alpha t)+ \rmi \sin{\alpha t}$.

In view of \eqref{asymp}, as $t\to \pm\infty$ the real vector solution $\vgamma^{(1)}$ starts and ends along the axis
$(1,0,0)$, although for finite $t$ its evolution can be complicated (see Figures 4 and 5 below).

\subsection{Examples of meromorphic solutions.}
Below we show two examples of independent meromorphic solutions $\{\vgamma_1, \vgamma_2, \vgamma_3\}$
of the Poisson equations for $p=1$ and $p=3$. 
First, the solutions \eqref{gam1}, \eqref{gam2} were taken, their real and imaginary
parts were extracted, then the formulae \eqref{eq:gam3} were applied.
We write the obtained {\it real} solutions in terms of hyperbolic and trigonometric functions.

The simples solutions for $p=1$ are
\[
 \begin{split}
 &
\vgamma^{(1)} =\left(\tanh(t),-\dfrac{\sech(t)}{\sqrt{1+d^2}},\dfrac{d\sech(t)}{
\sqrt{1+d^2}}\right),\\
&
\vgamma^{(2)}=\left(\cos\left(\frac{t}{d}\right)\sech(t),\dfrac{\cos(\frac{t}{d})\tanh(t)-d\sin(\frac{t}{d})}{\sqrt
{d^2+1}},-\dfrac{d\cos\left(\frac{t}{d}\right)\tanh(t)+\sin(\frac{t}{d})}{\sqrt{d^2+1}}\right),\\
&
\vgamma^{(3)}=\left(\sin\left(\frac{t}{d}\right)\sech(t),\dfrac{\sin(\frac{t}{d})\tanh(t)+d\cos(\frac{t}{d})}{\sqrt
{d^2+1}},-\dfrac{d\sin(\frac{t}{d})\tanh(t)-\cos(\frac{t}{d})}{\sqrt{d^2+1}}\right),
 \end{split}
\]
where $\sech (t)= 2/(e^t+e^{-t})$.

As one can check, these vectors form an orthonormal basis, $\langle \vgamma^{(i)}, \vgamma^{(j)}\rangle =\delta_{ij}$,
and, as $t\to \pm \infty$, we have $\vgamma^{(1)}=(\pm 1, 0,0)$ and, respectively,
\begin{align*}
\vgamma^{(2)} & = \left( 0, \frac { \cos (t/d) - d\sin (t/d) }{\sqrt{d^2+1}},
- \frac {d \cos (t/d) + \sin (t/d) }{\sqrt{d^2+1}} \right), \\
\vgamma^{(2)} & = \left( 0, -\frac {\cos (t/d) + d\sin (t/d) }{\sqrt{d^2+1}},
\frac {d \cos (t/d) - \sin (t/d) }{\sqrt{d^2+1}} \right).
\end{align*}
Since $d>0$, as $t \to -\infty$ (respectively $t \to -\infty$), in the body frame the vectors
$ \vgamma^{(2)}, \vgamma^{(3)}$ perform uniform rotations in the plane (0,1,1) in counterclockwise (resp. clockwise)
direction.

This implies the limit {\it spatial} motions of the body are rotations about the same axis,
with the same angular velocity $1/d$ and
in the {\it same} direction, however the rotation angle undergoes the phase shift, which is computed to
be $\arccos\frac{1-d^2}{1+d^2}$.
An example of the above solutions is illustrated in Fig. 4. 

\begin{figure}[h,t] \label{sp.fig}
\begin{center}
\includegraphics[height=0.4\textwidth]{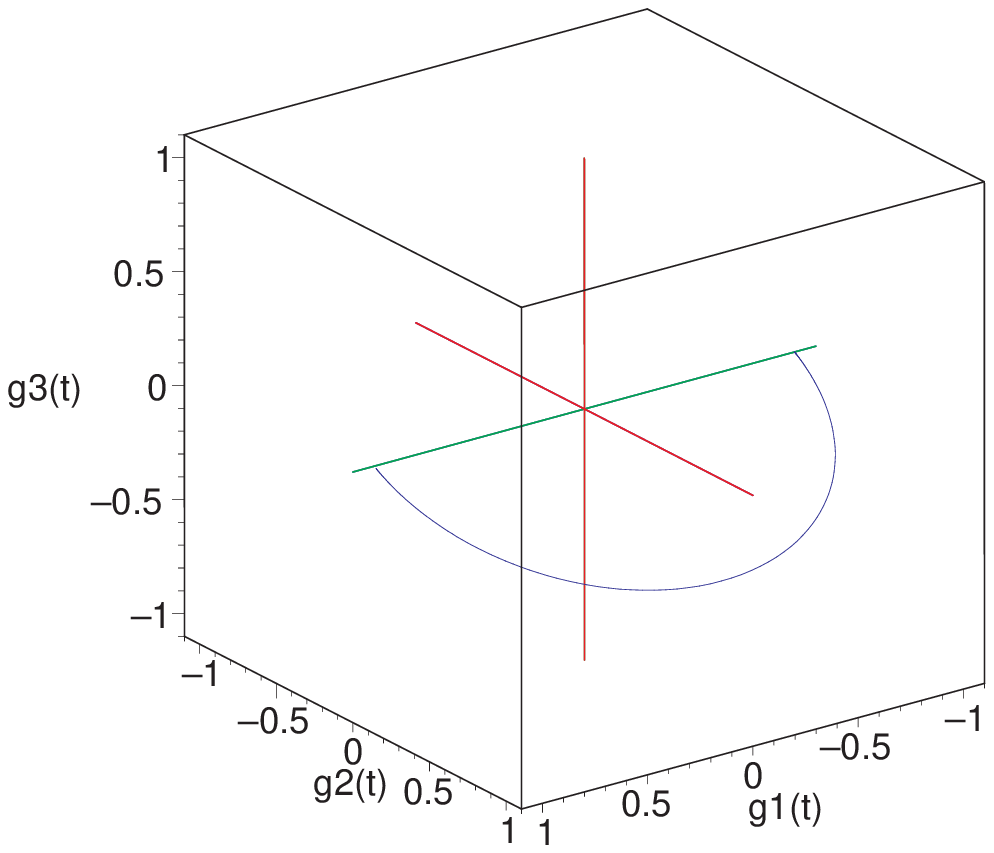} \quad
\includegraphics[height=0.4\textwidth]{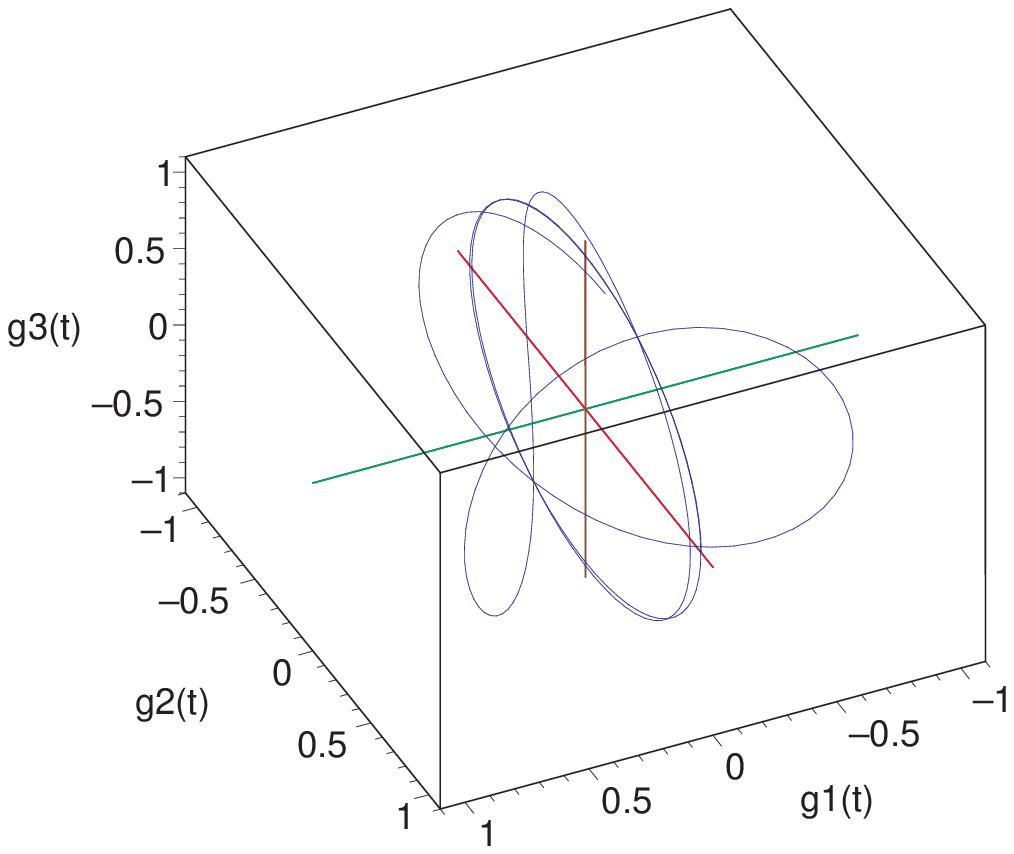}
\end{center}
\caption{\footnotesize Independent and orthogonal vector
solutions $\vgamma^{(1)}$ and $ \vgamma^{(2)}$ for $p=1$ and  $d=1/2$ in the body frame}
\end{figure}

For $p=3$ the corresponding solutions are more complicated:
\begin{align*}
 \gamma^{(1)} & =\alpha\Big( \sqrt{d^2+1}\Big(9 + d^2 - 4 d^2\sech(t)^2\Big)\tanh(t),\left[
4d^2\sech(t)^2-9(d^2+1)\right]\sech(t),\\
& \hskip 7cm d\left[3(d^2+1)-4d^2\sech(t)^2\right]\sech(t) \Big), \\
\gamma^{(2)} & = \dfrac{\alpha}{(e^{2 t}+1)^3}\Big(-2 \sqrt{d^2+1}\, e^t \Big[\Big(
       d^2 (3 - 10 e^{2 t} + 3 e^{4 t})-9 (1 + e^{2 t})^2\Big) \cos\left(\frac{3 t}{d}\right)\\
& \hskip 8cm + 12 d (e^{4 t}-1) \sin\left( \frac{3 t}{d}\right )\Big], \\
& \quad 4 e^{3 t}\Big[d \Big(( d^2-15) \cosh(2 t)+9 + d^2 \Big) \cosh(t) \sin\left(\frac{3 t}{d}\right) \\
& \hskip 4cm  +  \Big ((9 - 7 d^2) \cosh(2 t)+9 + d^2 \Big) \sinh(t) \cos\left(\frac{3 t}{d}\right)\Big], \\
& \qquad \qquad 4 e^{3 t} \Big[ \Big( (7 d^2-9) \cosh(2 t)-9 - 17 d^2 \Big)
\times \cosh(t) \sin\Big(\frac{3 t}{d}\Big)  \\
& \hskip 3cm  + d \Big( (d^2-15) \cosh(2 t)-15 - 7 d^2 \Big) \sinh(t)\cos\left(\frac{3 t}{d}\right)\Big] \Big),
 \end{align*}
and
\begin{align*}
\gamma^{(3)} & = \dfrac{\alpha}{(e^{2 t}+1)^3}\Big(2 \sqrt{d^2+1}\,e^t \Big[ \Big(
      d^2 (3 - 10 e^{2 t} + 3 e^{4 t})-9 (e^{2 t}+1)^2\Big) \sin\left(\frac{3 t}{d}\right)\\
& \hskip 8cm  +12 d (1-e^{4 t}) \cos\left(\frac{3 t}{d}\right)\Big], \\
& \quad 4 e^{3 t} \Big[d  \Big((d^2-15 ) \cosh(2 t)+9 + d^2 \Big)\cosh(t)\cos\left(\frac{3 t}{d}\right)\\
& \hskip 3cm + \Big((7d^2-9) \cosh(2 t)-9 - d^2\Big) \sinh(t)\sin\left(\frac{3 t}{d}\right)\Big], \\
& -4 e^{3 t} \Big[\Big((9 - 7 d^2) \cosh(2 t)+9 + 17 d^2 \Big) \cosh(t)\cos\left(\frac{3 t}{d}\right)\\
& \hskip 3cm + d \Big((d^2-15) \cosh(2 t)-15 - 7 d^2\Big) \sinh(t)\sin\left(\frac{3 t}{d}\right)\Big]
\Big),
 \end{align*}
where
\[
 \alpha=\dfrac{1}{\sqrt{d^2+1} (d^2+9)}.
\]
These vector solutions also form an orthonormal basis and the corresponding spatial motion of the body
is similar to that of the previous case $p=1$. An example of these solutions is illustrated in Fig. 5.

\begin{figure}[ht]
\centerline{\includegraphics[height=0.4\textwidth]{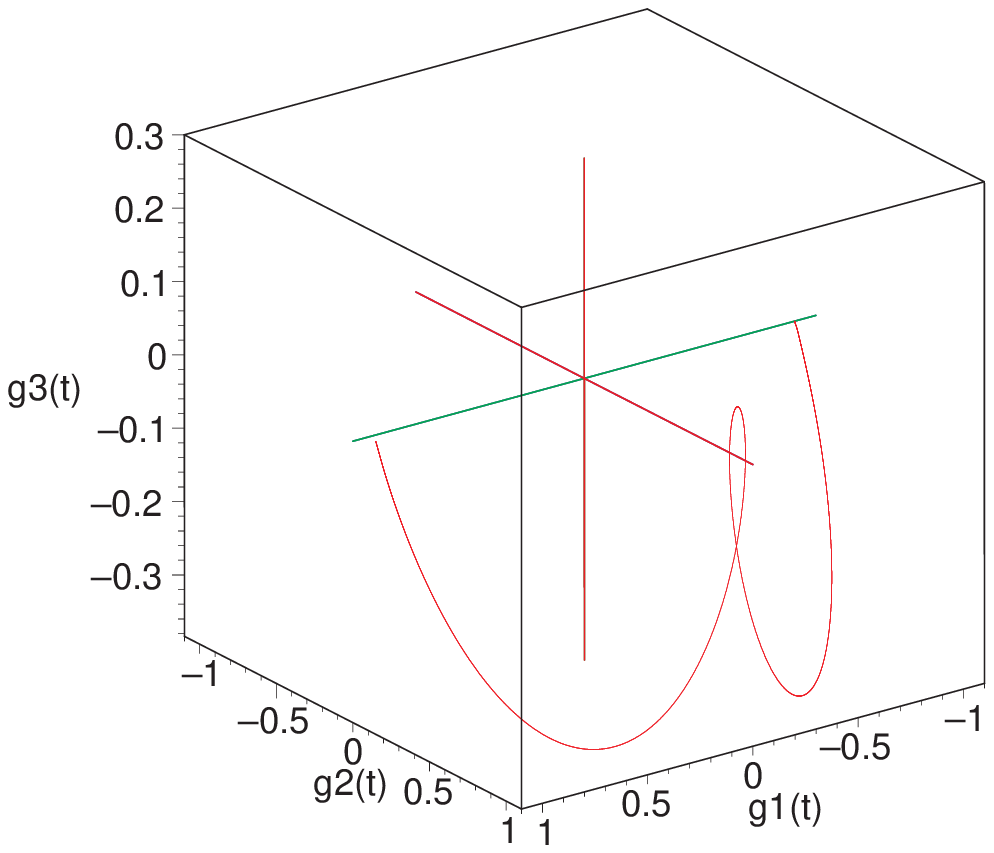}\hspace{0.7cm}
\includegraphics[height=0.4\textwidth]{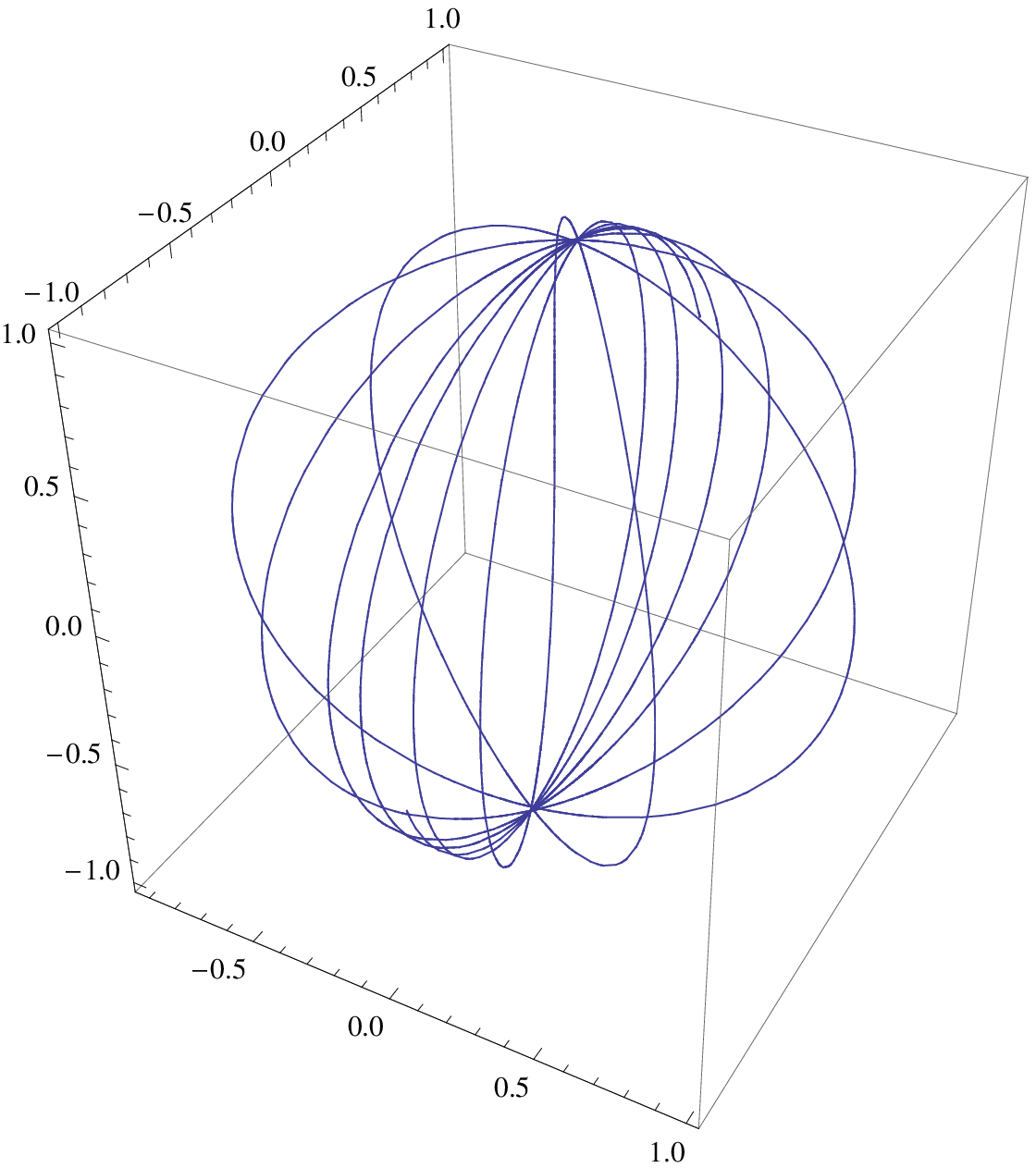}}
\caption{Vector solutions $\gamma^{(1)}$ and $ \gamma^{(2)}$ for $p=3$ and  $d=1/2$}
\end{figure}
\newpage

\subsection{First integrals.}
As was mentioned in the beginning of Section 2, if
$P(t)=(P_1,P_2,P_3)$ is a single-valued vector solution of the Poisson equations and, moreover,
the components of $P(t)$ are single-valued functions of $\omega_1(t), \omega_2(t)$, then
the system~\eqref{gen12}, \eqref{eq:whole} admits an additional first integral
\begin{equation}
 F_3 =P_1(\omega_1,\omega_2)\gamma_1+P_2(\omega_1,\omega_2)\gamma_2+P_3(\omega_1,\omega_2)\gamma_3 .
\label{eq:liniowiec}
\end{equation}
Since this system is homogeneous, then $P_i$ are homogeneous polynomials of a certain degree $k$.

Now we use the meromorphic solutions $\vgamma^{(1)}$ obtained in Section 6.1 for odd $p$ to construct
the corresponding extra integrals. In the simplest case $p=1$, comparing the components of $\vgamma^{(1)}$ with
the solutions \eqref{omegas} for $\omega_1, \omega_2$ and recalling the definition of $p,d$, we easily obtain
\begin{align*}
P_1 (t) & = \frac{\omega_1(t)}{a} = d\omega_1, \\
P_2(t) & = \frac{4 a \omega_1(t)}{b^2} = \frac{d}{d^2+1} \omega_2,  \\
P_3(t) & = - \frac{4 \omega_1(t)}{b^2} = - \frac{d^2}{d^2+1} \omega_2,
\end{align*}
which yields the integral
$$
F_3 = d \omega_1 \gamma_1 + \frac{d}{d^2+1} \omega_2 \gamma_2 - \frac{d^2}{d^2+1} \omega_2 \gamma_3 .
$$

In case of generic odd $p$ the extra integral can be written by means of
the special polynomial solution \eqref{pol_sol} of the generalized hypergeometric equation
\eqref{eq:genhip} with the parameters defined in \eqref{eq:param}.
To do this, it is convenient to define the following homogeneous function
\begin{equation}
 Q=Q(\omega_1,\omega_2) := F_1^{(p-1)/2}{ \cal F}_1 (z),
\end{equation}
where
\begin{equation}
 F_1 :=(d^2+1)\omega_1^2+\omega_2^2, \quad  \qquad z:=\frac{\omega_2^2}{F_1},
\end{equation}
and
${\cal F}_1(z)$ is the special polynomial solution \eqref{pol_sol} of \eqref{eq:genhip}, that is,
\begin{gather*}
{\cal F}_1(z) = 1+ \sum_{j=1}^{(p-1)/2} \frac{(2j-1)!!}{(2j)!!}\,
\frac {d^{2j} (p^2-1)\cdots (p^2-(2j-1)^2) }{(d^2+p^2)\cdots ((2j-1)d^2+p^2)}  \, z^j
\end{gather*}
The above definitions imply that if $p$ is odd natural number,
then $Q_1\in\R[\omega_1,\omega_2]$ is a homogeneous polynomial of degree $(p-1)$.
Moreover, $Q_1(\omega_1,\omega_2)$ is  an even function of $\omega_1$, as well as even function of $\omega_2$.

\begin{theorem}
If $p$ is odd natural number, then there exists  there homogeneous polynomials
$P_1,P_2,P_3\in\R[\omega_1,\omega_2]$ of the same degree $p$ such that
\eqref{eq:liniowiec} is a polynomial first integral of the Poisson equations.  Moreover
\begin{equation}
\label{eq:p1}
 P_1  = \omega_1  Q ,
\end{equation}
and
\begin{equation}
\begin{split}
P_3:=& \dfrac{d}{p}\left(\dfrac{1}{d^2+1}\dfrac{\partial P_1}{\partial
\omega_1}\omega_2-\dfrac{\partial P_1}{\partial
\omega_2}\omega_1\right),\\
P_2:=&- \dfrac{d}{p}\frac{\omega_2}{\omega_1}\left(\dfrac{1}{d^2+1}\dfrac{\partial P_3}{\partial
\omega_1}\omega_2-\dfrac{\partial P_3}{\partial
\omega_2}\omega_1\right) +\frac{\omega_2}{\omega_1} P_1 .
\end{split}
\label{eq:p23}
\end{equation}
 \end{theorem}
\begin{proof}
As follows from the definition of $P_1$, the expressions ~\eqref{eq:p23} for
$P_3$ and $P_2$ are polynomials. This is clear for $P_3$. Let us show that also $P_2\in\R[\omega_1,\omega_2]$.
To do this we observe that
\begin{equation}
\label{eq:dp3}
 \frac{\partial P_3}{\partial\omega_1}= -2pd\omega_2\frac{\partial Q}{\partial\omega_1}
- pd\omega_1\omega_2\frac{\partial^2 Q}{\partial\omega_1^2} + pd(1+d^2)\omega_1 \left[ 2\frac{\partial Q}{\partial\omega_2} +\omega_1 \frac{\partial^2 Q}{\partial\omega_1 \partial\omega_2}\right]
\end{equation}
Since $Q$ is an even function of $\omega_1$, we have
\begin{equation*}
 \frac{\partial Q}{\partial\omega_1} = \omega_1\widetilde Q, \quad\text{where} \quad
\widetilde Q \in \R[\omega_1,\omega_2],
\end{equation*}
and thus, from~\eqref{eq:dp3}, we have also
\begin{equation*}
 \frac{\partial P_3}{\partial\omega_1} =
\omega_1\widetilde P_3, \quad\text{where} \quad \widetilde P_3 \in \R[\omega_1,\omega_2],
\end{equation*}
Now from \eqref{eq:p23} it easily follows that $P_2\in\R[\omega_1,\omega_2]$.

It remains to show that with $P_1$, $P_2$ and $P_3$ given by~\eqref{eq:p1} and~\eqref{eq:p23},
the function~\eqref{eq:liniowiec} is an integral of the system.
Indeed, if it is a first integral, then its time derivative vanishes, i.e.,
\[
 0=\dot F_3=\dfrac{d}{p}\sum_{i=1}^3\left(\dfrac{1}{d^2+1}\dfrac{\partial
P_i}{\partial \omega_1}\omega_2-\dfrac{\partial P_i}{\partial
\omega_2}
\omega_1\right)\omega_2\gamma_i+P_3\omega_2\gamma_1-P_3\omega_1\gamma_2+
(P_2\omega_1-P_1\omega_2)\gamma_3 .
\]
The right hand side of the above equation is a linear form in $\gamma_i$,
hence we have the following system of three partial differential equations
\begin{equation}
\begin{split}
& \dfrac{d}{p}\left(\dfrac{1}{d^2+1}\dfrac{\partial P_1}{\partial
\omega_1}\omega_2-\dfrac{\partial P_1}{\partial
\omega_2}\omega_1\right)+P_3=0,\\
& \dfrac{d}{p}\left(\dfrac{1}{d^2+1}\dfrac{\partial P_2}{\partial
\omega_1}\omega_2-\dfrac{\partial P_2}{\partial
\omega_2}\omega_1\right)\omega_2-P_3\omega_1=0,\\
& \dfrac{d}{p}\left(\dfrac{1}{d^2+1}\dfrac{\partial P_3}{\partial
\omega_1}\omega_2-\dfrac{\partial P_3}{\partial
\omega_2}\omega_1\right)\omega_2+P_2\omega_1-P_1\omega_2=0.
\end{split}
\label{eq:partial}
\end{equation}
The form of this system shows that $P_2$ and $P_3$ can be expressed in terms $P_1$ and its partial derivatives.
The explicit form of these expression is given by~\eqref{eq:p23}.
Using them we eliminate $P_2$ and $P_3$ in \eqref{eq:partial} and obtain
one partial differential equation for $P_1$.
Next, taking into account the fact that $P_1$ is assumed to be a homogeneous polynomial,
we obtain an ordinary linear equation for $P_1$.
It remains to show that one of its solutions is given by formula~\eqref{eq:p1}.

Introduce a new variable $\omega=\omega_2/\omega_1$ and define
\[
 P_i (\omega_1,\omega_2) =\omega_1^pP_i(1, \omega_2/\omega_1)=:\omega_1^pp_i(\omega),\qquad i=1,2,3 .
\]
Then, we have
\[
 \dfrac{\partial P_i}{\partial \omega_1}=\omega_1^{p-1}[pp_i-\omega p_i'],\qquad
 \dfrac{\partial P_i}{\partial \omega_2}=\omega_1^{p-1}p_i',\qquad i=1,2,3,
\]
where prime denotes the derivative with respect to $\omega$.
This implies that our system of partial differential
equations~\eqref{eq:partial} on $P_i$ can be written as the system of ordinary
differential equations on $p_i$
\begin{equation}
 \begin{split}
 &-d(\omega^2+d^2+1)p_1'+dp\omega p_1+p(d^2+1)p_3=0,\\
&d(\omega^2+d^2+1)\omega p_2'-dp\omega^2 p_2+p(d^2+1)p_3=0,\\
&d(\omega^2+d^2+1)\omega p_3'-dp\omega^2 p_3+p(d^2+1)\omega p_1-p(d^2+1)p_2=0.
 \end{split}
\label{eq:ordin}
\end{equation}
Resolving this system with respect to $p_2, p_3$, we find
\begin{align}
 p_3 & =\dfrac{d(d^2+1+\omega^2)}{p(d^2+1)}p_1'-\dfrac{d}{(d^2+1)}\omega p_1, \nonumber \\
 p_2 & =\dfrac{d^2(1 + d^2 + \omega^2)^2\omega}{(1 + d^2)^2 p^2}p_1'' \nonumber \\
     & \qquad -\dfrac{2 d^2 (p-1) \omega^2 (1 + d^2 + \omega^2)}{(1 + d^2)^2 p^2}p_1'+\left(1-\dfrac{d^2 p
(1 + d^2 - (p-1) \omega^2)}{(1 + d^2)^2 p^2}\right)\omega p_1.
\label{p_23}
\end{align}
Substituting the above expression  into the second equation in \eqref{eq:ordin}, we obtain the following third  order linear equation
\begin{align}
& \qquad p_1'''+q_1p_1''+q_2p_1'+q_3p_1=0,  \label{eq:syfint} \\
q_1 & =\dfrac{1}{\omega} - \dfrac{3 ( p-2) \omega}{1 + d^2 + \omega^2}, \notag \\
q_2 & =\dfrac{d^2 \omega^2 (-(1 + d^2) (5 p-4 ) + (p-1) (3 p-8) \omega^2) + (1 +
      d^2)^2 (1 + \omega^2) p^2}{d^2 \omega^2 (1 + d^2 + \omega^2)^2}, \notag \\
q_3 & =-\dfrac{1}{d^2 \omega (1 + d^2 + \omega^2)^3}\Big[d^2 p ((1 + d^2)^2 -
       4 (1 + d^2) (p-1) \omega^2 + (3 + (p-4) p) \omega^4)\notag\\
& - (1 +
       d^2)^2 (d^2 - (p-1) (1 + \omega^2)) p^2\Big]. \notag
\end{align}
Now, making change of independent variable to
\[
 z=\dfrac{\omega^2}{\omega^2+d^2+1} \equiv \frac{4}{(e^t+e^{-t})^2},
\]
and using
\[
 \begin{split}
 \dfrac{\mathrm{d} z}{\mathrm{d}\omega}=\dfrac{2 \omega (z-1)^2}{d^2+1},\quad
\dfrac{\mathrm{d}^2 z}{\mathrm{d}\omega^2}=-\dfrac{2 (z-1)^2 (4 z-1)}{d^2+1},\quad
\dfrac{\mathrm{d}^3 z}{\mathrm{d}\omega^3}=-\dfrac{24 \omega ( z-1)^3 (2 z-1)}{(d^2+1)^2},
 \end{split}
\]
we transform the equation \eqref{eq:syfint} into
\begin{gather*}
p_1'''+b_1p_1''+b_2p_1'+b_3p_1=0,  \qquad
{}'\equiv\dfrac{\mathrm{d}\phantom{-}}{\mathrm{d}z} \\
b_1=\dfrac{2}{z}+\dfrac{2 + 3 p}{2 (z-1)},\\
b_2= \dfrac{
 d^2 (1 + ( p^2-6 - 8 p) z + (4 + 3 p (2 + p) -
       p^2) z^2)-p^2 (z-1) }{4 d^2 (z-1)^2 z^2},\\
b_3=\dfrac{p^2 (-1 + p + z - p z) +
 d^2 (p^2 ( z-1) - 4 p^2 z + p^3 z^2 +
    p (1 + (2 - p^2 ( z-1)) z))}{8 d^2 ( z-1)^3 z^2}.
\end{gather*}
Now under the change of the dependent variable
\begin{equation}
 p_1(z)=(z-1)^{\frac{1-p}{2}} v(z),
\label{eq:chan}
\end{equation}
the equation for $v(z)$ becomes exactly the generalized hypergeometric equation
\eqref{eq:genhip} defining the function $\phantom{\vert}_3F_2$  with the parameters
\eqref{eq:param}.
On the other hand, it easy to notice that, up to a multiplicative constant,
the dehomogenized $P_1$ given by~\eqref{eq:p1} and expressed in terms of the variable
$z$ coincides with~\eqref{eq:chan}.
\end{proof}

\subsection{Conclusion} As follows from the results of Sections 3 and 5,
when $I_{13}=0$ and the parameter $p$ us an even integer, an interesting situation takes place:
all the solutions of the Poisson equations are meromorphic, but the equations itself are not solvable in the
class of Liouvillian functions. In particular, they do not possess extra meromorphic first integrals. In this case
it is natural to expect that the Poisson equations are reducible to one of the Painlev\'e equations.

It also should be emphasized that we managed to reduce the Poisson equations to the
hypergeometric equation under the restriction $I_{13}=0$. The study of solutions of the third order equation
\eqref{eq:genhipprep} in the general case, especially of its monodromy group, is an interesting open problem.

\section*{Acknowledgements} The first author (Yu.F.) acknowledges the
support of grant MTM 2006-14603 of Spanish Ministry of Science and Technology.

The research of the second and third authors (A. M and M.P) was supported by grant No. N N202 2126 33
of Ministry of Science and Higher Education of Poland.

The research of M.P. was also partially supported by
Projet de l'Agence National de la Recherche ``Int\'egrabilit\'e r\'eelle et
complexe en m\'ecanique hamiltonienne'' N$^\circ$~JC05$_-$41465 and by the grant UMK 414-A.

\end{document}